\documentclass[journal,twoside,web]{ieeecolor}
\usepackage{generic}
\usepackage{cite}
\usepackage{amsmath,amssymb,amsfonts}
\usepackage{mathrsfs}
\usepackage{algorithm}
\usepackage{algorithmic}
\usepackage{graphicx}
\usepackage{textcomp}
\usepackage{hyperref}
\newtheorem{theorem}{Theorem}
\newtheorem{definition}{Definition}
\newtheorem{corollary}{Corollary}
\newtheorem{lemma}{Lemma}
\newtheorem{fact}{Fact}
\newtheorem{proposition}{Proposition}
\newtheorem{assumption}{Assumption}

\let\labelindent\relax
\usepackage{mathtools}
\usepackage{enumitem}
\mathtoolsset{showonlyrefs}
\usepackage{tikz}

\usetikzlibrary{shapes}
 \usetikzlibrary{calc}
\usetikzlibrary{decorations.pathreplacing}
\usetikzlibrary{arrows,positioning} 
\usetikzlibrary{pgfplots.groupplots}
\usetikzlibrary{decorations.text}

\tikzstyle{link}=[line width=2pt, ->,>=latex]
\tikzstyle{junc}=[draw,circle,inner sep=1pt,minimum width=8pt]
\tikzstyle{onramp}=[line width=2pt, dashed,->,>=latex]

\DeclareMathOperator*{\argmin}{arg\,min}
\DeclareMathOperator*{\argmax}{arg\,max}

\makeatletter
\DeclareRobustCommand\widecheck[1]{{\mathpalette\@widecheck{#1}}}
\def\@widecheck#1#2{%
    \setbox\z@\hbox{\m@th$#1#2$}%
    \setbox\tw@\hbox{\m@th$#1%
       \widehat{%
          \vrule\@width\z@\@height\ht\z@
          \vrule\@height\z@\@width\wd\z@}$}%
    \dp\tw@-\ht\z@
    \@tempdima\ht\z@ \advance\@tempdima2\ht\tw@ \divide\@tempdima\thr@@
    \setbox\tw@\hbox{%
       \raise\@tempdima\hbox{\scalebox{1}[-1]{\lower\@tempdima\box
\tw@}}}%
    {\ooalign{\box\tw@ \cr \box\z@}}}
\makeatother

\newcommand{\F}{\mathscr{F}}

\def\BibTeX{{\rm B\kern-.05em{\sc i\kern-.025em b}\kern-.08em
    T\kern-.1667em\lower.7ex\hbox{E}\kern-.125emX}}
\begin{document}
\title{Specification-guided Verification and Abstraction Refinement of Mixed Monotone Stochastic Systems}
\author{Maxence Dutreix and Samuel Coogan
\thanks{Manuscript submitted on ...}
\thanks{This project was supported in part by the NSF under project \#1749357.}
\thanks{M. Dutreix is with the School of Electrical and Computer Engineering, Georgia Institute of Technology, Atlanta, GA, USA (e-mail: maxdutreix@gatech.edu).}
\thanks{S. Coogan is with the School of Electrical and Computer Engineering and the School of Civil and Environmental Engineering, Georgia Institute of Technology, Atlanta, GA, USA (e-mail: sam.coogan@gatech.edu).}}

\maketitle

\begin{abstract}
This paper addresses the problem of verifying discrete-time stochastic systems against omega-regular specifications using finite-state abstractions. Omega-regular properties allow specifying complex behavior and encompass, for example, linear temporal logic. We focus on a class of systems with mixed monotone dynamics. This class has been show to be amenable to efficient reachable set computation and models a wide-range of physically relevant systems. In general, finite-state abstractions of continuous state stochastic systems give rise to augmented Markov Chains wherein the probabilities of transition between states are restricted to an interval. We present a procedure to compute a finite-state Interval-valued Markov Chain abstraction of discrete-time, mixed-monotone stochastic systems subject to affine disturbances given a rectangular partition of the state-space. Then, we suggest an algorithm for performing verification against omega-regular properties in IMCs. Specifically, we aim to compute bounds on the probability of satisfying a specification from any initial state in the IMC. This is achieved by solving a reachability problem on sets of so-called winning and losing components in the Cartesian product between the IMC and a Rabin automaton representing the specification. Next, the verification of IMCs may yield a set of states whose acceptance status is undecided with respect to the specification, requiring a refinement of the abstraction. We describe a specification-guided approach that compares the best and worst-case behaviors of accepting paths in the IMC and targets the appropriate states accordingly. Finally, we show a case study.
\end{abstract}


\section{Introduction}

Recent years have witnessed a growing effort to design tools for verifying and optimizing the behavior of systems with respect to increasingly complex specifications. Many efficient verification techniques can be readily applied in the context of purely deterministic models. Yet, the study of stochastic systems generally requires a machinery of its own.  For example, the model checking problem for discrete-time and continuous-time \textit{Markov Chains} (MC) has been solved for the \textit{Probablistic Computation Tree Logic} (PCTL) and \textit{$\omega$-regular properties} \cite{baier2008principles} \cite{aziz2000model} \cite{baier2003model}. Numerous off-the-shelf tools, such as PRISM \cite{kwiatkowska2011prism}, can efficiently verify MCs for a wide range of specifications.

However, these MC models assume that the probabilities of transition between states are known exactly. While computationally convenient, such representations may not be realistic in practice. In particular, the verification of continuous state systems is often performed by partitioning its domain into a discrete set of states. Because these states abstract a collection of continuous behaviors simultaneously, modeling a stochastic finite-state abstraction as a transition system can be challenging. One approach uses approximate Markov chains \cite{abate2011approximate} as in the FAUST${}^2$ model checker \cite{soudjani2014faust}, where each discrete state is reduced to a single representative point. Another approach uses Markov set-chains \cite{abate2008markov} to abstract the evolution of such systems so that existing techniques are applicable. Similar automata-based techniques are found in \cite{abate2011quantitative} and \cite{tkachev2013formula} for linear time objectives. However, these techniques generally rely on a fine gridding of the continuous state-space which overlooks qualitative aspects that are important to certain specifications, e.g. the creation of absorbing states as pointed out in \cite{abate2011quantitative}. Abstraction-based methods have been employed as well for the verification of continuous-time stochastic hybrid models against safety specifications in \cite{hahn2013compositional} and \cite{franzle2011measurability}.

Instead of approximate Markov chains, finite-state abstractions of continuous state-space, discrete-time stochastic systems are also amenable to Markovian models where the probabilities of transitions are restricted to some interval \cite{lahijanian2015formal} \cite{dutreix2018}, referred to as \textit{Interval-valued Markov Chains} (IMC) \cite{kozine2002interval}. We consider this approach in this work. Constructing interval-valued abstractions of stochastic systems with continuous state-spaces is often a computationally expensive process. Indeed, in the general case, calculating the exact lower and upper bounds of transitions between all discrete states involves numerical searches over the state-space, rendering this procedure highly time-inefficient. Nevertheless, we aim to show that this impractical computational blowup can be avoided in some cases by exploiting the inherent structure of the system's dynamics. Our proposed approach results in an IMC abstraction with conservative transition probability ranges, which remains sufficient for verification.

In particular, we consider a class of stochastic systems for which the dynamics exhibits a \emph{mixed monotone} property \cite{smith2008global, coogan2015efficient}. Mixed monotonicity generalizes the property of monotonicity for dynamical systems for which trajectories maintain a partial ordering on states \cite{Hirsch:1985fk}, \cite{Smith:2008fk}, \cite{Angeli:2003fv}. Many physical systems have been shown to be monotone or mixed monotone such as biological systems \cite{Sontag:2007ad} and transportation networks \cite{Gomes:2008fk}, \cite{Lovisari:2014yq}, \cite{Coogan:2016rp}. This property enables efficient computation of reachable sets: a rectangular over-approximation of the one-step reachable set from any rectangular discrete state is determined by evaluating a certain decomposition function at the least and the greatest point of that state. In this paper, we study mixed monotone systems that are subject to an affine random disturbance vector whose components are mutually independent and for which the probability distribution for each component is unimodal and symmetric. For such systems, we show that an upper bound and a lower bound on the probability of transitions between states of a rectangular partition are found by evaluating only two integrals per dimension and per transition. We use these bounds to create an IMC abstraction of the original system that is suitable for verification.

 Efficient verification algorithms for IMCs have been derived for the logic PCTL \cite{lahijanian2015formal}. However, many specifications of interest cannot be expressed using PCTL, such as liveness, i.e., the infinitely repeated occurrence of an event \cite{rozier2011linear}. This motivates the implementation of a machinery for handling such properties for the class of abstractions mentioned above. As a superset of the broadly used \textit{Linear Temporal Logic} (LTL), $\omega$-regular properties exhibit considerable expressiveness. 
 
%

In this work, we treat the problem of verification for discrete-time, continuous-space stochastic systems against $\omega$-regular properties. To this end, we use finite state abstractions in the form of IMCs. We develop an algorithm for determining the best and worst-case outcome of the IMC by solving a reachability problem in the product between the IMC and a Rabin automaton corresponding to the specification of interest. Our approach can be decomposed into a graph search which identifies the \textit{largest} and \textit{permanent} so-called \textit{winning} and \textit{losing components} \cite{baier2004controller} in the product IMC, and the computation of bounds on the probability of reaching these components.

Such verification techniques may yield a set of states whose satisfiability status with respect to the desired specification is undecided. Previous works suggested a methodology to compute a gridding parameter that guarantees an upper bound on the size of the interval of satisfaction for all states in the resulting partition and for all specifications in the logic PCTL \cite{abate2011approximate}. This parameter is a function of the system's properties only, such as Lipschitz constants. However, this method is often conservative and likely to provide very fine partitions and is therefore computationally expensive. Instead, we apply a specification-guided refinement procedure on an initial crude partition of the state-space and iteratively produce finer IMC abstractions until some precision criterion is attained. Verification is performed on each new partition and a specific set of states is selected for the next refinement step. Refinement heuristics were proposed for the PCTL framework in \cite{lahijanian2015formal} and \cite{dutreix2018}. We present a technique for the $\omega$-regular framework that accounts for the behavior of the paths generated by the best and worst case adversary of the IMC to refine regions of the state-space with the highest potential of reducing uncertainty. 

In summary, the first contribution of this paper is an efficient procedure for constructing an IMC abstraction of stochastic mixed-monotone systems subject to affine, unimodal and symmetrical disturbances. The second contribution is a technique for computing satisfiability bounds on $\omega$-regular specifications in IMCs. The third contribution is a specification-guided refinement method that selectively and iteratively partitions certain regions of the continuous domain of stochastic systems for the purpose of verification against $\omega$-regular properties. This technique improves on other specification-guided approaches by inspecting the accepting paths produced by the worst and best case adversary of the IMC instead of observing only the one-step transitions of each state and by considering the qualitative structure of the product IMC through its largest and permanent components to target the states accordingly.

The paper is organized as follows: Section II introduces preliminaries; Section III presents the problems investigated in this paper; Section IV derives an IMC abstraction techniques for affine-in-disturbance mixed monotone systems;  Section V discusses the verification of IMCs against $\omega$-regular properties; Section VI discusses state-space refinement; Section VII contains case studies; Section VIII summarizes the conclusions.

Preliminary results were reported in the conference papers \cite{dutreix2018} and \cite{dutreixCDC2018}. The verification approach considered here is a significant improvement of these prior works. Specifically, the methodology proposed here does not require computing the complement of the system specification and presents a more advanced algorithm for partition refinement.

%
\IEEEpeerreviewmaketitle

\begin{figure}[t]
\setlength{\belowcaptionskip}{-8pt}
\centering
\includegraphics[scale=0.30]{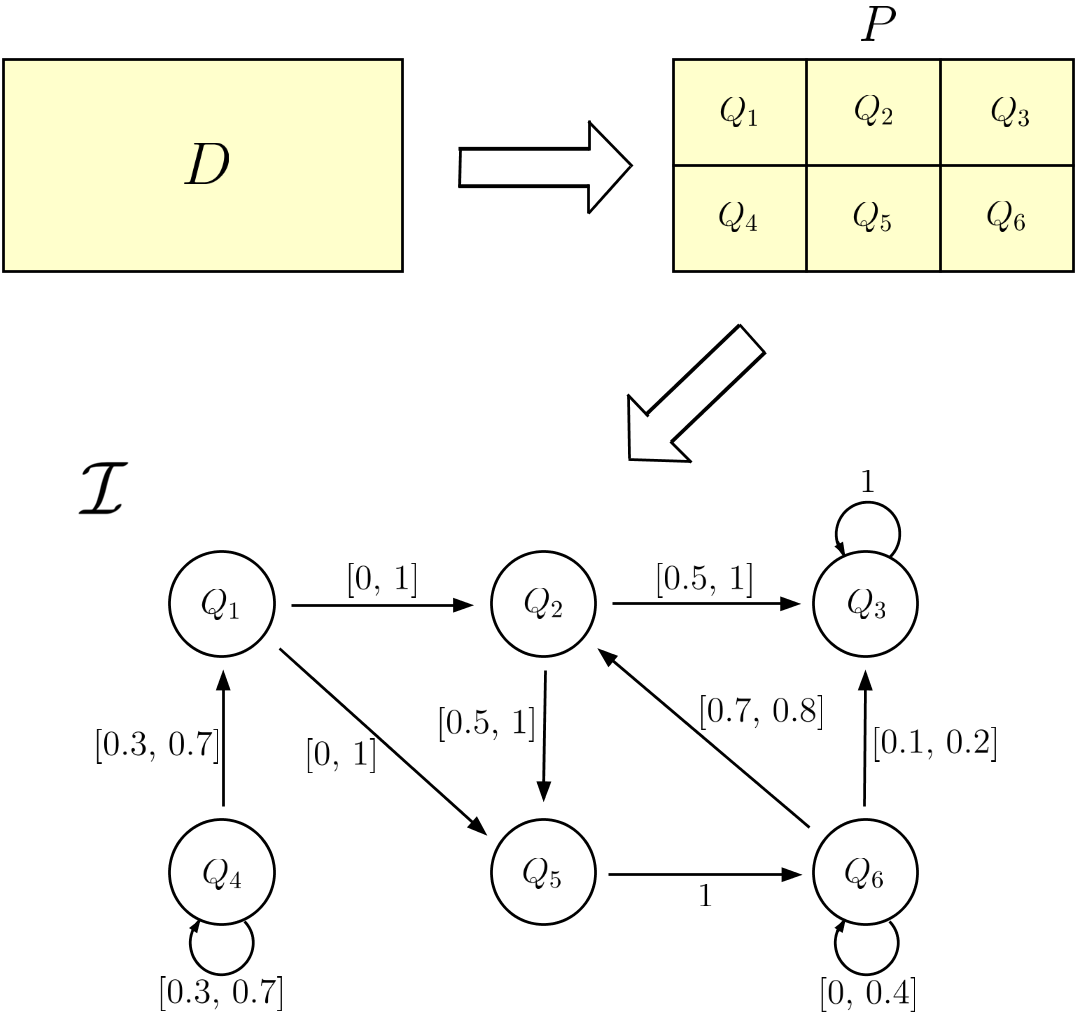}
\label{figabs}
\vspace{-0.2cm}
\caption{A finite state IMC abstraction $\mathcal{I}$ over a continuous domain $D$. A partition $P$ of $D$ is generated and bounds on the transition probabilities between states are estimated.}
\end{figure}

\section{Preliminaries}


%
A \textit{Deterministic Rabin Automaton (DRA)} \cite{baier2008principles} is a 5-tuple $\mathcal{A} = (S, \Sigma, \delta, s_0, Acc)$ where:
\begin{itemize}
\item $S$ is a finite set of states,
\item $\Sigma$ is an alphabet,
\item $\delta : S \times \Sigma \rightarrow S$ is a transition function
\item $s_0$ is an initial state
\item $Acc \subseteq 2^{S} \times 2^{S}$. An element $(E_{i}, F_{i}) \in Acc$, with $E_{i}, F_{i} \subset S$, is called a \textit{Rabin Pair}.
\end{itemize}

\noindent A DRA $\mathcal{A}$ reads an infinite string or \textit{word} over alphabet $\Sigma$ as an input and transitions from state to state according to $\delta$. The resulting sequence of states or \textit{run} is an \textit{accepting} run if it contains an infinite number of states belonging to $F_i$ and a finite number of states in $E_i$ for some $i$. A word is said to be accepted by $\mathcal{A}$ if it produces an accepting run in $\mathcal{A}$. We call a set of words a \textit{property}. The property \textit{accepted} by $\mathcal{A}$ is the set of all words accepted by $\mathcal{A}$.

A property over an alphabet $\Sigma$ is \textit{$\omega$-regular} if and only if it is accepted by a Rabin Automaton with alphabet $\Sigma$ (for more detailed definitions of $\omega$-regular properties, see \cite[Section 4.3.1]{baier2008principles}). In particular, all properties defined by a \textit{Linear Temporal Logic} (LTL) formula are $\omega$-regular. See \cite{baier2008principles} for a description of the semantics of LTL.

%
%

Throughout, all inequalities are interpreted elementwise so that, for $x, y \in \mathbb{R}^{n}$, $x \leq y$ means $x_{i} \leq y_{i}$ for $i = 1, \ldots
, n$, and similarly for $\geq, <$ and $>$. 

$Q_{j} = \lbrace x: a^{j} \leq x \leq b^{j} \rbrace$ for some $a^{j}, b^{j} \in \mathbb{R}^{n}$ such that $a^{j} \leq b^{j}$ is a \textit{compact rectangular set} and $a^{j}$, $b^j$  are respectively called the \textit{least point} and the \textit{greatest point} of $Q_{j}$. For vectors, we reserve the subscript to index elements of the vector so that, \emph{e.g.}, $a^j_i$ for $i\in\{1,\ldots,n\}$ denotes the $i$-th element of $a^j\in\mathbb{R}^n$.

\section{Problem Formulation}

We consider the discrete-time, continuous-state stochastic system
\begin{align}
\label{eq1}
x[k+1] = \mathcal{F}(x[k], w[k])
\end{align}
where $x[k] \in D \subset \mathbb{R}^{n}$ is the state of the system at time $k$, $w[k] \in W \subset \mathbb{R}^{p}$ is a random disturbance and $\mathcal{F}: D \times W \rightarrow D$ is a continuous map. At each time-step $k$, the random disturbance $w[k]$ is sampled from a probability distribution with density function $f_w:\mathbb{R}^p\to \mathbb{R}_{\geq 0}$ satisfying $f_w(z)=0$ if $z\not\in W$. Let $L: D \rightarrow \Sigma$ be a labeling function, where $\Sigma$ is a finite alphabet. A random path $x[0] x[1] \ldots$ satisfying \eqref{eq1} generates the word $L(x[1]) L(x[2]) \ldots$ over $\Sigma$.

We denote by $\Psi$ an $\omega$-regular property over alphabet $\Sigma$ and define a probability operator $\mathcal{P}_{\bowtie p_{sat}}[ \Psi ]$ over $\omega$-regular properties, with $\bowtie \; \in \lbrace \leq, <, \geq, > \rbrace$, $p_{sat} \in [0, 1]$. For any initial state $x \in D$, we define the satisfaction relation $\models$ where
\begin{align*}
x \models \mathcal{P}_{\bowtie p_{sat}}[ \Psi ] \Leftrightarrow p^x_{\Psi} \bowtie p_{sat}\; ,
\end{align*}
\noindent with $p^x_{\Psi}$ being the probability that the word generated by a random path starting in $x$ satisfies property $\Psi$ (for a rigorous formalization, see, e.g., \cite{abate2011approximate}). In this paper, we concentrate on formulas of the type
\begin{align}
\label{eq2}
\phi = \mathcal{P}_{\bowtie p_{sat}}[ \Psi ] \;.
\end{align}
 
\noindent Our objective is to sort all initial states of system \eqref{eq1} into those that satisfy and those that do not satisfy specification \eqref{eq2}, denoted by the sets $Q^{yes}_{\phi}$ and $Q^{no}_{\phi}$ respectively.\\

\textbf{Problem}: \textit{Given a system of the form \eqref{eq1}, find the sets of initial states $Q^{yes}_{\phi} \subseteq D$ and $Q^{no}_{\phi}  \subseteq D$ that respectively satisfy and do not satisfy a formula $\phi$ of the form \eqref{eq2}.}\\

The domain $D$ of system \eqref{eq1} generally contains an uncountably infinite number of states and obtaining exact solutions to this problem may be infeasible for rich specifications. A common approximation approach consists in partitioning $D$ into a finite collection of states $P$ to obtain a finite abstraction of the stochastic dynamics. In this paper, we only consider partitions which are rectangular.\\

\begin{definition}[Rectangular Partition]
A \emph{rectangular partition} $P$ of the compact domain $D \subset \mathbb{R}^{n}$ is a collection of discrete states $P = \lbrace Q_{j} \rbrace_{j=1}^{m}, \; Q_{j} \subset D,$ satisfying
\begin{itemize}
\setlength{\itemsep}{0pt}
\item $Q_{j}$ is a compact rectangular set $\forall j = 1, \ldots, m$,
\item $\bigcup_{j=1}^{m} Q_{j} = D$,
\item $\text{int}(Q_{j}) \cap \text{int}(Q_{\ell}) = \emptyset \;\; \forall j, \ell, \;  j \not = \ell \;\;,$
\end{itemize}
\noindent where int denotes the interior.
For any continuous state $x$ belonging to a state $Q_j$, we write $x \in Q_j$.\\
\end{definition}
\noindent Henceforth, we assume that $D$ in system (1) is amenable to a rectangular partition.

Given a rectangular partition $P$ for a system (1), the likelihood of transitioning from a state $Q_j$ of $P$ to another state $Q_{\ell}$ generally varies with the continuous state abstracted by $Q_{j}$ from which the transition is actually taking place. This prevents using the partition $P$ to uniquely abstract \eqref{eq1} as a finite discrete-time MC. Instead, we produce an IMC abstraction of the system where the transition probabilities between states are constrained within some bounds. Fig. 1 depicts a schematic of a partition-based finite abstraction.\\

\begin{definition}[Interval-Valued Markov Chain]
An \textit{Interval-Valued Markov Chain (IMC)} \cite{dutreix2018} is a 5-tuple $\mathcal{I} = (Q, \widecheck{T}, \widehat{T}, \Sigma, L)$ where:
\begin{itemize}
\item $Q$ is a finite set of states,
\item $\widecheck{T}: Q \times Q \rightarrow [0, 1] $ maps pairs of states to a lower transition bound so that $\widecheck{T}_{Q_{j} \rightarrow Q_{\ell}} := \widecheck{T}(Q_{j}, Q_{\ell})$ denotes the lower bound of the transition probability from state $Q_{j}$ to state $Q_{\ell}$, and 
\item $\widehat{T}: Q \times Q \rightarrow [0, 1] $ maps pairs of states to an upper transition bound so that $\widehat{T}_{Q_{j} \rightarrow Q_{\ell}} := \widehat{T}(Q_{j}, Q_{\ell})$ denotes the upper bound of the transition probability from state $Q_{j}$ to state $Q_{\ell}$,
\item $\Sigma$ is a finite set of atomic propositions,
\item $L : Q \rightarrow \Sigma$ is a labeling function from states to $\Sigma$,
\end{itemize}
and $\widecheck{T}$ and $\widehat{T}$ satisfy $\widecheck{T}(Q_j,Q_\ell)\leq \widehat{T}(Q_j,Q_\ell)$ for all $Q_j,Q_\ell\in Q$ and 
\begin{equation}
\label{eq:36}
\sum_{Q_\ell\in Q} \widecheck{T}(Q_j,Q_\ell)\leq 1\leq \sum_{Q_\ell\in Q}\widehat{T}(Q_j,Q_\ell)  
\end{equation}
 for all $Q_j\in Q$.
 \end{definition}\mbox{}
 
\begin{definition}[Markov Chain]
A \textit{Markov Chain} (MC) $\mathcal{M} = (Q, T, \Sigma, L)$ is defined similarly to an IMC with the difference that the transition probability function or transition matrix $T: Q \times Q \rightarrow [0, 1] $ satisfies $0 \leq T(Q_j, Q_\ell) \leq 1$ for all $Q_j, Q_{\ell}\in Q$ and $\sum_{Q_{\ell} \in Q} T(Q_j, Q_\ell) = 1$ for all $Q_j\in Q$.
\end{definition}\mbox{}
 
A Markov Chain $\mathcal{M} = (Q, T, \Sigma, L)$ is said to be \textit{induced} by IMC $\mathcal{I} = (Q, \widecheck{T}, \widehat{T}, \Sigma, L)$ if they share the same $Q$, $\Sigma$ and $L$, and for all $Q_j,Q_\ell\in Q$, $\widecheck{T}(Q_j,Q_\ell) \leq T(Q_j, Q_\ell) \leq \widehat{T}(Q_j,Q_\ell).$ A transition matrix $T$ satisfying this inequality is said to be induced by $\mathcal{I}$.\\

\noindent The notation $\mathcal{P}_{\mathcal{M}}(Q_i\models \Diamond U)$ for $U \subseteq Q$ denotes the probability of eventually reaching set $U$ from initial state $Q_i$ in Markov Chain $\mathcal{M}$. With a slight abuse of notation, when the probability of reaching $U$ is the same from all states in a set $G$, we write it as $\mathcal{P}_{\mathcal{M}}(G \models \Diamond U)$.\\

An IMC $\mathcal{I}$ is interpreted as an \textit{Interval Markov Decision Process} (IMDP) \cite{sen2006model} if, at each time step $k$, the environment non-deterministically chooses a transition matrix $T_k$ induced by $\mathcal{I}$ and the next transition occurs according to $T_k$. A mapping $\mathcal{\nu}$ from a finite path $\pi = q_0 q_1\ldots q_k$ in $\mathcal{I}$ to a transition matrix $T_k$ is called an \textit{adversary}. The set of all adversaries of $\mathcal{I}$ is denoted by $\mathcal{\nu}_{\mathcal{I}}$.

The probability of satisfying $\omega$-regular property $\Psi$ starting from initial state $Q_i$ in IMC $\mathcal{I}$ under adversary $\mathcal{\nu}$ is denoted by $\mathcal{P}_{\mathcal{I}[\mathcal{\nu}]}(Q_i \models \Psi)$.
The greatest lower bound and least upper bound on the probability of satisfying property $\Psi$ starting from initial state $Q_i$ in IMC $\mathcal{I}$ are denoted by $\widecheck{\mathcal{P}}_{\mathcal{I}}(Q_i \models \Psi) = \inf_{\nu \in \nu_{\mathcal{I}}} \mathcal{P}_{\mathcal{I}[\mathcal{\nu}]}(Q_i \models \Psi)$ and $\widehat{\mathcal{P}}_{\mathcal{I}}(Q_i \models \Psi) = \sup_{\nu \in \nu_{\mathcal{I}}} \mathcal{P}_{\mathcal{I}[\mathcal{\nu}]}(Q_i \models \Psi)$ respectively.\\

\begin{definition}[IMC Abstraction]
Given the system \eqref{eq1} evolving on a domain $D \subset  \mathbb{R}^{n}$ and a partition $P=\{Q_j\}_{j=1}^m$ of $D$, an IMC $\mathcal{I}=(Q, \widecheck{T}, \widehat{T}, \Sigma, L)$ is an \emph{abstraction} of \eqref{eq1} if:  
\begin{itemize}[leftmargin=1.1em]
\item $P=Q$, i.e., the set of states of the IMC is the partition $P$, 
\item For all $Q_j,Q_\ell\in P$,
\begin{align}
\label{eq:3}   \widecheck{T}_{Q_{j} \rightarrow Q_{\ell}}  &\leq  \inf_{x\in Q_j} Pr( \mathcal{F} (x,w) \in Q_\ell ),\text{ and}\\
\label{eq:3-2}   \widehat{T}_{Q_{j} \rightarrow Q_{\ell}}  &\geq \sup_{x\in Q_j} Pr( \mathcal{F} (x,w) \in Q_\ell ),
\end{align}
where $Pr( \mathcal{F}(x,w) \in Q_\ell)$ for fixed $x$ is the probability that \eqref{eq1} transitions from $x$ to a state $x'=F(x,w)$ in $Q_{\ell}$,
\item For all $Q_j \in P$ and for any two states $x_i, x_{\ell} \in Q_j$, it holds that $L(Q_j) := L(x_i) = L(x_{\ell})$, that is, the partition respects the boundaries induced by the labeling function,
\item $\mathcal{I}$ is interpreted as an IMDP.
\end{itemize}\mbox{}
\end{definition}

The fact that two continuous states within the same discrete state of the abstraction may engender different transition probabilities is captured by interpreting the IMC as an IMDP.

\noindent One can always build a trivial IMC abstraction where all transitions are set to range from 0 to 1. The ease of finding tighter bounds on the transitions between states is dictated by the dynamics of the system and the geometry of the states in the partition. A contribution of this paper is to show that \textit{affine-in-disturbance stochastic mixed monotone systems} evolving on a domain amenable to a rectangular partition can efficiently be abstracted by IMCs. A formal definition of such systems is provided in the next section.\\

\textbf{Subproblem 1}: \textit{Given an affine-in-disturbance stochastic mixed monotone system, construct a non-trivial IMC abstraction from a rectangular partition of its domain.}\\

Performing verification on an IMC abstraction provides probabilistic guarantees with respect to the original system's states. A consequence of model checking an IMC $\mathcal{I}$ is that the probability of satisfying property $\Psi$ from any of its initial states $Q_j$ must be specified as an interval $I_{j} = [p^{j}_{min}, p^{j}_{max}]$, where $\mathcal{P}_{\mathcal{I}[\mathcal{\nu}]}(Q_j \models \Psi) \in I_{j}, \; \forall \nu \in \nu_{\mathcal{I}}$. For any initial state $Q_j$ in an IMC $\mathcal{I}$, we define the satisfaction relation $\models$ for formulas of the type \eqref{eq2} where
\begin{align*}
Q_j \models \mathcal{P}_{\bowtie p_{sat}}[ \Psi ] \Leftrightarrow (p^{Q_j}_{\Psi})_{\nu} \bowtie p_{sat}\;\; \forall \nu \in {\nu_{\mathcal{I}}} ,
\end{align*} 
\noindent with $(p^{Q_j}_{\Psi})_{\nu}$ being the probability that the word generated by a random path starting in $Q_j$ satisfies property $\Psi$ under adversary $\nu$. We denote the set of initial states satisfying $\phi$ in $\mathcal{I}$ by $(Q^{yes}_{\phi})_{\mathcal{I}}$, while states that do not satisfy $\phi$ are in $(Q^{no}_{\phi})_{\mathcal{I}}$.  Note that any $Q_{j}$ such that $p_{sat} \in \; ]p^{j}_{min}, p^{j}_{max}[$ if $\bowtie \; \in \lbrace \leq, \geq\rbrace$, or $p_{sat} \in [p^{j}_{min}, p^{j}_{max}]$ if $\bowtie \; \in \lbrace <,  > \rbrace$, is undecided with respect to $\phi$ in \eqref{eq2} and we write $Q_{j} \in (Q^{?}_{\phi})_{\mathcal{I}}$. The remaining states either satisfy $\phi$ or do not satisfy $\phi$.\\

\begin{fact}
Let $\mathcal{I}$ be an IMC abstraction of \eqref{eq1} induced by a partition $P=\{Q_j\}_{j=1}^m$ of $D$. For any formula of the form \eqref{eq2}, it holds that:
\begin{itemize}
\item $Q_{j} \in (Q^{yes}_{\phi})_{\mathcal{I}} \Rightarrow x \in Q^{yes}_{\phi} \;\; \forall x \in Q_{j}$
\item $Q_{j} \in (Q^{no}_{\phi})_{\mathcal{I}} \; \Rightarrow x \in Q^{no}_{\phi} \;\;\; \forall x \in Q_{j} \; .$
\end{itemize}
\end{fact}\mbox{}

\noindent Given an IMC abstraction $\mathcal{I}$ of \eqref{eq1} generated from a partition $P$ of $D$, our approach for addressing the Verification Problem is thus to implement a technique for determining non-trivial values of $p^{j}_{min}$ and $p^{j}_{max}$ and sort all states of $P$ into the sets $(Q^{yes}_{\phi})_{\mathcal{I}}, (Q^{no}_{\phi})_{\mathcal{I}}$ and $(Q^{?}_{\phi})_{\mathcal{I}}$.\\

\textbf{Subproblem 2}: \textit{Given an IMC $\mathcal{I}$ interpreted as an IMDP and an $\omega$-regular property $\Psi$, find the greatest lower bound $\widecheck{\mathcal{P}}_{\mathcal{I}}(Q_j \models \Psi)$ and the least upper bound  $\widehat{\mathcal{P}}_{\mathcal{I}}(Q_j \models \Psi)$ on the probability of satisfying $\Psi$ from any initial state $Q_j$ in $\mathcal{I}$.}\\


It may be the case that the total volume of uncertain states in $Q^{?}_{\phi}$ is unsatisfactorily large due to a crude partition $P$ of the continuous state-space. In such event, we aim to produce a finer partition $P'$ from $P$ which gives rise to a new IMC abstraction of the system with a greater number of states with tighter transition bounds between them. Verifying this finer IMC yields a decreased volume of uncertain states. This refinement procedure is applied to subsequent partitions until some discretionary level of precision is reached.\\

\textbf{Subproblem 3}: \textit{Given a system of the form \eqref{eq1} with an IMC abstraction $\mathcal{I}$ of (1) and the corresponding sets $Q^{yes}_{\phi}$, $Q^{no}_{\phi}$ and $Q^{?}_{\phi}$ obtained by solving Subproblem 2, refine the partition $P$ of $D$ until a predefined threshold of precision has been met.}\\

\noindent In order to avoid state-space explosion, a careful choice of the states to be refined in partition $P$ has to be made. In this work, we suggest a novel method for targeting the states that have the greatest potential of reducing the uncertainty of an abstraction with respect to the particular specification $\Psi$.

\section{IMC Abstraction of Mixed Monotone Systems Affine in Disturbance}

In this section, we study a large class of stochastic systems that proves amenable to efficient computation. In particular, we consider affine-in-disturbance systems of the form
\begin{align}
x[k+1] = \mathcal{F}(x[k]) + w[k]
\label{eq:4}
\end{align}
with specific assumptions on $\mathcal{F}$ and the distribution of the disturbance $w$. We introduce several definitions and specify these necessary assumptions for addressing Subproblem 1.\\

\begin{definition}[Mixed monotone function]
A function $\mathcal{F}: D \rightarrow D$ is \emph{mixed monotone} if there exists a \emph{decomposition function} $g: D \times D \rightarrow D$ satisfying \cite{smith2008global, coogan2015efficient}:
\begin{itemize}
\setlength{\itemsep}{0pt}
\item $\forall x \in D:  \mathcal{F}(x) = g(x,x)$
\item $\forall x^{1}, x^{2}, y \in D: x^{1} \leq x^{2}$  implies $g(x^{1}, y) \leq g(x^{2}, y)$
\item $\forall x, y^{1}, y^{2} \in D: y^{1} \leq y^{2}$ implies $g(x, y^{2}) \leq g(x, y^{1})$
\end{itemize}
\end{definition}\mbox{}

Mixed monotonicity generalizes the notion of monotonicity in dynamical systems, which is recovered when $g(x,y)=\mathcal{F}(x)$ for all $x,y$. Systems with monotone state update maps exhibit considerable structure useful for analysis \cite{Hirsch:2005ek, Angeli:2003fv, Smith:2008fk}. Systems with mixed monotone state update maps have been shown to enjoy many of these same structural properties \cite{coogan2015efficient, smith2008global}. For example, for mixed monotone $\mathcal{F}$ with decomposition function $g$, for $x,y, z\in D$ satisfying $x\leq z\leq y$, we have $g(x,y)\leq \mathcal{F}(z) \leq g(y,x)$. This leads to the following proposition.\\

\begin{proposition}[{\cite[Theorem 1]{coogan2015efficient}}]
\label{prop:reach}
  Let $\mathcal{F}: D \rightarrow D$ be mixed monotone with decomposition function $g: D \times D \rightarrow D$, and let $a,b\in D$ satisfy $a\leq b$. Then
  \begin{align}
  \{\F(x):a\leq x\leq b \}\subseteq \{z: g(a,b)\leq z\leq g(b,a)\}.    
  \end{align}
\end{proposition}\mbox{}

Proposition \ref{prop:reach} implies that the one-step reachable set from the rectangular region bounded between $a$ and $b$ is overapproximated by the rectangular region  bounded by the two points $g(a,b)$ and $g(b,a)$. This property will prove key for efficient computation of IMC abstractions.\\

\begin{definition}[Unimodal symmetric distribution]
For a random disturbance $\omega \in \Omega \subset \mathbb{R}$ with $\Omega$ an interval, its probability density function $f_{\omega} : \mathbb{R} \rightarrow \mathbb{R}$ is \emph{unimodal} if $f_{\omega}$ is differentiable on $\Omega$ and  there exists a unique number $c \in \mathbb{R}$, referred as the mode of the distribution, such that, for $x\in \Omega$:
\begin{itemize}
\item $x < c \Rightarrow f_{\omega}'(x) \geq 0$,
\item $x = c \Rightarrow f_{\omega}'(x) = 0$, and
\item $x > c  \Rightarrow f_{\omega}'(x) \leq 0$.
\end{itemize}
\noindent We only consider distributions without a ``flat'' peak, that is, unimodal distributions with a unique mode $c$. The probability density function $f_{\omega}$ is \emph{symmetric} if there exists a number $d \in \mathbb{R}$ such that $f_{\omega}(d-x) = f_{\omega}(d+x)$ for all $ x$.
\end{definition}\mbox{}

Note that if $f_\omega$ is unimodal with mode $c$ and symmetric, then it must be that $f_{\omega}(c-x) = f_{\omega}(c+x)$.

Henceforth, we make the following assumptions. \\

\begin{assumption}
\label{assum:1}
  $\mathcal{F}(x)$ in \eqref{eq:4} is mixed monotone with decomposition function $g(x,y)$. Furthermore, the domain $D$ of \eqref{eq:4} can be partitioned into a rectangular partition.
\end{assumption}\mbox{}

\begin{assumption}
\label{assum:indep}
  The random disturbance $w[k]$ in \eqref{eq:4} is of the form $w[k]= \begin{bmatrix}w_{1}[k]& w_{2}[k]& \ldots& w_{n}[k]\end{bmatrix}^{T}$, where each $w_{i} \in W_{i} \subset \mathbb{R}$ has probability density function $f_{w_i}(x_i)$, $W_i$ is an interval, and the collection $\{w_i\}_{i=1}^n$ is mutually independent. Denote by $F_{w_i}(x)=\int_{-\infty}^x f_{w_i}(\sigma) d\sigma$ the cumulative distribution function for $w_i$. 
\end{assumption}\mbox{}

\begin{assumption}
\label{assum:uni}
The probability density function $f_{w_i}$ for each random variable $w_{i}$ is symmetric and unimodal with mode $c_i$.
\end{assumption}\mbox{}

We now address Subproblem 1 for systems of the form \eqref{eq:4} under Assumptions 1-3. We decompose our procedure for bounding the transition probability from a state $Q_{1}$ to a state $Q_{2}$ in two steps: first, we compute the rectangular over-approximation of the $\mathcal{F}$-reachable set from state $Q_{1}$ by taking advantage of the mixed monotonicity property. Next, we determine the positions of $f_w$ within this rectangular region that respectively minimize and maximize its overlap with $Q_{2}$. In the next section, we exploit the characteristics of $w$ previously evoked to obtain an efficient computational procedure for computing these extremum points.\\

\begin{proposition}
\label{prop:1}
Consider system \eqref{eq:4} under Assumptions \ref{assum:1}--\ref{assum:uni}. Let $Q_1=\{x: a^1\leq x\leq b^1\}$ and $Q_2=\{x:a^2\leq x\leq b^2\}$ be two nonempty rectangular sets with least point $a^j$ and greatest point $b^j$ for $j=1,2$. Then
\begin{align}
\nonumber  &\min_{x\in Q_1} Pr(\mathcal{F}(x)+w\in Q_2)\\
\label{eq:2}&\geq \prod_{i=1}^n\min_{\substack{z_i\in[ g_i(a^1,b^1), g_i(b^1,a^1)]}}\int_{a_i^2}^{b_i^2}f_{w_i}(x-z_i) dx 
\end{align}
and
\begin{align}
\nonumber  &\max_{x\in Q_1} Pr(\mathcal{F}(x)+w\in Q_2)\\
\label{eq:5}&\leq \prod_{i=1}^n\max_{\substack{z_i\in[ g_i(a^1,b^1), g_i(b^1,a^1)]}}\int_{a_i^2}^{b_i^2}f_{w_i}(x-z_i) dx 
\end{align}
where $g_i$ denotes the $i$-th element of $g(x,y)$, the decomposition function of $\mathcal{F}$.
\end{proposition}

All proofs are found in the appendix. Before generalizing to higher dimensions, we treat a 1-dimensional version of our original problem. In Lemma 1, we prove that for a fixed interval $[a, b] \subset\mathbb{R}$, there exists a unique position for a unimodal and symmetric distribution which maximizes its integral over $[a, b]$.\\

\begin{figure}
\centering
\includegraphics[scale=0.3, trim={0 5.5cm 0 5cm}, clip]{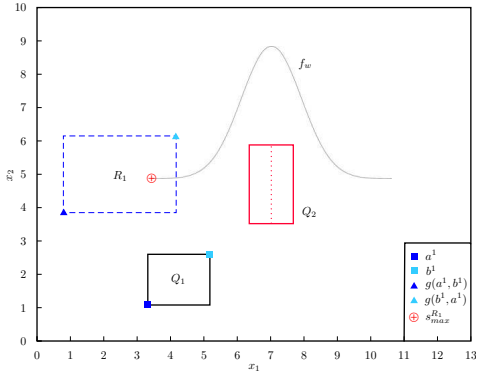}
\caption{Depiction of the procedure for computing an upper
bound on the probability of transition from $Q_1$ to $Q_2$. First, the one-step reachable set $R_1$ from $Q_1$ is over-approximated by evaluating the
decomposition function at only two extremal points. Then, the distribution of $z+w$ is positioned as close to the center of $Q_2$ as possible under the restriction that $z \in R_1$. A lower bound on the transition probability is achieved by positioning the distribution as far from the center of $Q_2$ as possible.}
\end{figure}

\begin{lemma}
\label{lem:main}
Let $\omega \in \Omega \subset \mathbb{R}$ with $\Omega$ an interval be a random variable with symmetric and unimodal probability density function $f_{\omega} : \mathbb{R} \rightarrow \mathbb{R}$ and mode $c\in\mathbb{R}$. For any $a,b\in\mathbb{R}$ satisfying $a\leq b$ and any $r_1,r_2\in\mathbb{R}$ satisfying $r_1\leq r_2$, let
\begin{align}
  s_{max} = \frac{a+b}{2} - c
\end{align}
and define
\begin{align}
  \label{eq:9}
s^{r}_{max}&=\underset{s \in [r_{1}, r_{2}]}{\argmin} |s_{max}- s| =\begin{cases}
s_{max},&\text{if} \;\; s_{max} \in [r_1,r_2] \\
r_{2}, &\text{if} \;\; s_{max} > r_{2}\\
r_{1},& \text{if} \;\; s_{max}< r_1,
\end{cases}\\
\label{eq:11} s^{r}_{min} &=\underset{s \in [r_{1}, r_{2}]}{\argmax} |s_{max} - s| =\begin{cases}
r_{1},&\text{if} \;\; s_{max} < \frac{r_{1}+r_{2}}{2} \\
r_{2}, &\text{otherwise}.
\end{cases}
\end{align}
Then
\begin{align}
\label{eq:6}
\underset{s \in [r_1,r_2]}{\max}\int_{a}^{b} f_{\omega}(x-s) \; dx&= \int_{a}^{b} f_{\omega}(x-s^{r}_{max}) \; dx\\
\label{eq:10}\underset{s \in  [r_1,r_2]}{\min}\int_{a}^{b} f_{w}(x-s) \; dx &= \int_{a}^{b} f_{\omega}(x-s^{r}_{min}) \; dx.
\end{align}





\end{lemma}\mbox{}

When $s_{max}\in[r_1,r_2]$ in Lemma \ref{lem:main}, the lemma confirms the intuitive idea that the integral of a unimodal, symmetric distribution over some interval $I=[a,b]$ is maximized when the peak of its probability  distribution lies at the center of $I$. However, for the type of systems considered in this work, the shift of such distributions will always be restricted to take values within a given rectangular set $[r_1,r_2]$ so that, when $s_{max}\not \in [r_1,r_2]$, the shift $s\in[r_1,r_2]$ maximizing the overlap of the density function over $I$ is the one closest to the global maximizing shift $s_{max}$. Conversely, a shift $s\in[r_1,r_2]$ minimizing this overlap is the one furthest from $s_{max}$.

Theorem 1 combines Lemma \ref{lem:main} and Proposition \ref{prop:1} to provide a procedure for constructing an IMC abstraction for \eqref{eq:4}  given a rectangular partition of its domain $D$.\\

\begin{theorem}
\label{thm:1}
Consider system \eqref{eq:4} under Assumptions \ref{assum:1}--\ref{assum:uni} and let $P=\{Q_j\}_{j=1}^m$ be a rectangular partition of $D$ with each $Q_j=\{x: a^j\leq x\leq b^j\}$ for some $a^j,b^j\in\mathbb{R}^n$ satisfying $a^j\leq b^j$. 
For all $Q_j,Q_\ell\in P$, let
\begin{align}
\label{eq:18}    s^\ell_{i,max}&=\frac{a^\ell_i+b^\ell_i}{2}-c_i\text{ for $i=1,\ldots,n$},\\
\label{eq:19}\widecheck{r}^j&=g(a^j,b^j),\\
\label{eq:20} \widehat{r}^j&=g(b^j,a^j),
\end{align}
and define
\begin{align}
\label{eq:15}\widehat{T}_{Q_{j} \rightarrow Q_{\ell}} & = \prod_{i = 1}^{n} \int_{a^{\ell}_{i}}^{b^{\ell}_{i}} f_{w_{i}}(x_{i}-s_{i,max}^{j \rightarrow \ell})\; dx_{i}, \\
& = \prod_{i = 1}^{n} \bigg(F_{w_{i}}(b^{\ell}_{i}-s_{i,max}^{j \rightarrow \ell}) - F_{w_{i}}(a^{\ell}_{i}-s_{i,max}^{j \rightarrow \ell})\bigg),\\
\label{eq:16}\widecheck{T}_{Q_{j} \rightarrow Q_{\ell}} & = \prod_{i = 1}^{n} \int_{a^{\ell}_{i}}^{b^{\ell}_{i}} f_{w_{i}}(x_{i}-s_{i,min}^{j\rightarrow \ell})\; dx_{i} \\ 
& = \prod_{i = 1}^{n} \bigg(F_{w_{i}}(b^{\ell}_{i}-s_{i,min}^{j \rightarrow \ell}) - F_{w_{i}}(a^{\ell}_{i}-s_{i,min}^{j \rightarrow \ell})\bigg)
\end{align}
where $F_{w_i}$ is the cumulative distribution function for $w_i$ and
\begin{align}
\label{eq:13}
  s_{i,max}^{j \rightarrow \ell}&=\begin{cases}
  s^\ell_{i,max},&\text{if} \;\;   s^\ell_{i,max} \in [\widecheck{r}_i^j, \widehat{r}_i^j] \\
\widehat{r}^j_i, &\text{if} \;\;  s^\ell_{i,max} > \widehat{r}^j_i\\
\widecheck{r}^j_i,& \text{if} \;\;  s^\ell_{i,max} < \widecheck{r}^j_i,
\end{cases}\\
\label{eq:14}  s_{i,min}^{j \rightarrow \ell}&=\begin{cases}
\widecheck{r}^j_i,&\text{if} \;\; s^\ell_{i,max} < \frac{\widecheck{r}^j_i+\widehat{r}^j_i}{2} \\
\widehat{r}^j_i, &\text{otherwise}.
\end{cases}
\end{align}
Then $\mathcal{I}=(P, \widecheck{T},\widehat{T})$ is an IMC abstraction of \eqref{eq:4}.

\end{theorem}\mbox{}


Theorem 1 provides the mathematical foundation for solving Subproblem 1. Given a system of the form $\eqref{eq:4}$ satisfying Assumptions 1 to 3, and a rectangular partition $P$, this theorem shows that an IMC abstraction of \eqref{eq:4} can be computed efficiently. Specifically, for any state in $P$, we establish an over-approximation of its one-step reachable set by evaluating the system's decomposition function at only two points. Finding the maximizing and minimizing shifts inside the reachable sets decouples along each coordinate and involves a number of operations that is linear in the dimension $n$ of the state-space, according to $\eqref{eq:13}$ and $\eqref{eq:14}$. Finally, we see in  $\eqref{eq:15}$ and $\eqref{eq:16}$ that $n$ integral evaluations are needed per transition bound. This last step requires two evaluations of $F_{w_i}$ per bound for each $i$ and thus amounts to $2n$ function evaluations per bound. Because the computed over-approximations of reachable sets for mixed monotone dynamics were shown to be tight in certain cases \cite{coogan2015efficient}, our abstraction technique is generally not as conservative as those employing Lipschitz constants \cite{cauchi2019efficiency}, while being computationally efficient.


\section{Verification of IMCs}

In this section, we develop the machinery to address Subproblem 2. Let $\mathcal{I}$ be an IMC abstraction of \eqref{eq1} obtained from, e.g., the abstraction approach in Section IV for systems with specific form \eqref{eq:4}. For a formula $\phi$ of the type \eqref{eq2}, our goal is to sort the initial states of $\mathcal{I}$ into the sets $Q^{yes}_{\phi}, Q^{no}_{\phi}$ and $Q^{?}_{\phi}$. To this end, for any initial state $Q_{j}$ of $\mathcal{I}$, we require a lower bound and an upper bound on the probability of satisfying the $\omega$-regular property $\Psi$ for the probabilistic specification $\phi=\mathcal{P}_{\bowtie p_{sat}}[\Psi]$. We thus seek to compute the greatest lower bound $\widecheck{\mathcal{P}}_{\mathcal{I}}(Q_j \models \Psi)$ and least upper bound $\widehat{\mathcal{P}}_{\mathcal{I}}(Q_j \models \Psi)$ such that, for any adversary $\nu \in \nu_{\mathcal{I}}$,
\begin{align}
\widecheck{\mathcal{P}}_{\mathcal{I}}(Q_j \models \Psi) \leq \mathcal{P}_{\mathcal{I}[\nu]}(Q_j \models \Psi) \leq \widehat{\mathcal{P}}_{\mathcal{I}}(Q_j \models \Psi) \; .
\end{align} 

Our approach draws from the verification of regular MCs against $\omega$-regular properties using automata-based methods \cite[Section 10.3]{baier2008principles}. First, we generate a DRA $\mathcal{A}$ that recognizes the language induced by property $\Psi$. Such a DRA always exists. Several algorithms can generate a DRA for a large subset of $\omega$-regular expressions \cite{klein2006experiments} \cite{babiak2013effective}. Then, we construct the product $\mathcal{I} \otimes \mathcal{A}$ as defined below.\\

\begin{definition}[Product Interval-valued Markov Chain]
Let $\mathcal{I} = (Q, \widecheck{T}, \widehat{T}, \Sigma, L)$ be an Interval-valued Markov Chain and $\mathcal{A} = (S, \Sigma, \delta, s_0, Acc)$ be a Deterministic Rabin Automaton. The \textit{product} $\mathcal{I} \otimes \mathcal{A} = (Q \times S, \widecheck{T'}, \widehat{T'}, Acc', L')$ is an Interval-valued Markov Chain where:
\begin{itemize}
\setlength{\itemsep}{0pt}
\item $Q \times S$ is a set of states,\\
\item $\{(Q_j,s_0):Q_j\in Q\}$ is a finite set of initial states,\\
\item $\widecheck{T'}_{ \left<Q_{j},s\right> \rightarrow \left<Q_{\ell},s'\right>} = 
\begin{cases}
\widecheck{T'}_{Q_{j} \rightarrow Q_{\ell}}, \;\; \text{if} \;\; s' = \delta(s, L(Q_{\ell}))\\ \;\;\; \;\;\;\;\;\;0, \;\;\;\;\;\; \text{otherwise}
\end{cases}$\mbox{}\\\\
\item $\widehat{T'}_{ \left<Q_{j},s\right> \rightarrow \left<Q_{\ell},s'\right>} = 
\begin{cases}
\widehat{T'}_{Q_{j} \rightarrow Q_{\ell}}, \;\; \text{if} \;\; s' = \delta(s, L(Q_{\ell}))\\ \;\;\; \;\;\;\;\;\;0, \;\;\;\;\;\; \text{otherwise}
\end{cases}$\mbox{}\\
\item $Acc' = \lbrace E_{1}, E_{2}, \ldots, E_{k}, F_{1}, F_{2}, \ldots, F_{k} \rbrace$ is a set of atomic propositions, where $E_{i}$ and $F_{i}$ are the sets in the Rabin pairs of $Acc$,
\item $L': Q \times S \rightarrow 2^{Acc'}$ such that $H \in L'(\left<Q_{j},s\right>)$ if and only if $s \in H$, for all $H \in Acc' $ and for all $j$.
\end{itemize}\mbox{}
\end{definition}

\noindent A MC induced by $\mathcal{I} \otimes \mathcal{A}$ is called a \textit{product Markov Chain}, and we use the notation $\mathcal{M}_{\otimes}^{\mathcal{A}}$ to denote such an induced MC. 

The probability of satisfying $\Psi$ from initial state $Q_j$ in a discrete-time MC equals that of reaching an accepting \textit{Bottom Strongly Connected Component} (BSCC) from initial state $\left<Q_j,s_0 \right>$ in the product MC with $\mathcal{A}$ \cite{baier2008principles}.\\

\begin{definition}[Bottom Strongly Connected Component]
Given a Markov Chain $\mathcal{M}$ with states $Q$, a set $B \subseteq Q$ is a \textit{Bottom Strongly Connected Component} (BSCC) of $\mathcal{M}$ if
\begin{itemize}
\item $B$ is strongly connected: for each pair of states $(q,t)$ in $B$, there exists a path $q_{0}q_{1}\ldots q_n$ such that $T(q_i,q_{i+1}) > 0$, $i = 0,1, \ldots, n-1$, and $q_i \in B$ for $0 \leq i \leq n$ with $q_0 = q,$ $q_n = t$,
\item no proper superset of $B$ is strongly connected,
\item $\forall s \in B$, $\Sigma_{t \in B}T(s,t) = 1$.
\end{itemize}
\end{definition}\mbox{}

In words, every state in a BSCC $B$ is reachable from any state in $B$, and every state in $B$ only transitions to another state in $B$. Moreover, $B$ is accepting when at least one of its states maps to the accepting set of a Rabin pair, while no state in $B$ maps to the non-accepting set of that same pair.\\

\begin{definition}[Accepting Bottom Strongly Connected Component]
A Bottom Strongly Connected Component $B$ of a product MC $\mathcal{M}_{\otimes}^{\mathcal{A}}$ is said to be \textit{accepting} if:
\begin{align}
\exists i: & \Big( \; \exists \left<Q_{j},s_{\ell} \right> \in B \; . \; F_{i} \in L'(\left<Q_{j},s_{\ell} \right>) \; \Big) \nonumber \\ 
& \wedge \Big( \;  \forall \left<Q_{j},s_{\ell} \right> \in B \; . \; E_{i} \not \in L'(\left<Q_{j},s_{\ell} \right>) \; \Big).
\end{align}
\end{definition}\mbox{}

\begin{definition}[Non-Accepting Bottom Strongly Connected Component]
A Bottom Strongly Connected Component $B$ of a product MC $\mathcal{M}_{\otimes}^{\mathcal{A}}$ is said to be \textit{non-accepting} if it is not accepting.
\end{definition}\mbox{}

\noindent We denote by $U^A$ and $U^N$ the sets of states that respectively belong to an accepting and a non-accepting BSCC in a product MC.

Note that each product MC $\mathcal{M}_{\otimes}^{\mathcal{A}}$ induced by $\mathcal{I} \otimes \mathcal{A}$ simulates the behavior of $\mathcal{I}$ under some adversary $\mathcal{\nu} \in \mathcal{\nu}_{\mathcal{I}}$. Indeed, for any two states $Q_j$ and $Q_{\ell}$ in $\mathcal{I}$ and some states $s, s', s''$ and $s'''$ in $\mathcal{A}$, we allow $T_{ \left<Q_{j},s\right> \rightarrow \left<Q_{\ell},s'\right>}$ and $T_{ \left<Q_{j},s''\right> \rightarrow \left<Q_{\ell},s'''\right>}$ to assume different values in $\mathcal{M}_{\otimes}^{\mathcal{A}}$, which means that the transition probability between $Q_j$ and $Q_{\ell}$ may change depending on the history of the path in $\mathcal{I}$ as encoded in the state of $\mathcal{A}$.

Also, the adversary is history-independent or \textit{memoryless} in the product automaton, that is, the adversary's chosen transition probability only depends on the current states of the IMC and the DRA $\mathcal{A}$. For reachability problems in IMCs, it was shown in \cite{chen2013complexity} that memoryless adversaries yield the same bounds as the memory-dependent ones. The following facts establish that, therefore, such memoryless (in the product) adversaries are sufficient for IMC verification.\\

\begin{fact}
\cite[p. 792, Theorem 10.56]{baier2008principles} \cite{chen2013complexity}  We denote the set of adversaries of $\mathcal{I}$ that are memoryless in the product IMC $\mathcal{I} \otimes \mathcal{A}$ by $(\mathcal{\nu}_{\mathcal{I}})_{\otimes}^{\mathcal{A}} \subseteq \mathcal{\nu}_{\mathcal{I}}$. It holds that 
\begin{align}
\inf_{\nu \in \nu_{\mathcal{I}}} \mathcal{P}_{\mathcal{I}[\mathcal{\nu}]}(Q_i \models \Psi) & = \inf_{\nu \in (\mathcal{\nu}_{\mathcal{I}})_{\otimes}^{\mathcal{A}}} \mathcal{P}_{\mathcal{I}[\mathcal{\nu}]}(Q_i \models \Psi)\\
\sup_{\nu \in \nu_{\mathcal{I}}} \mathcal{P}_{\mathcal{I}[\mathcal{\nu}]}(Q_i \models \Psi) &  =  \sup_{\nu \in (\mathcal{\nu}_{\mathcal{I}})_{\otimes}^{\mathcal{A}}} \mathcal{P}_{\mathcal{I}[\mathcal{\nu}]}(Q_i \models \Psi) \ .
\end{align}
\end{fact}\mbox{}

\begin{fact}
For any adversary $\mathcal{\nu} \in (\mathcal{\nu}_{\mathcal{I}})_{\otimes}^{\mathcal{A}}$ in $\mathcal{I}$, it holds that  $\mathcal{P}_{\mathcal{I}[\nu]}(Q_i \models \Psi)$ =  $\mathcal{P}_{(\mathcal{M}_{\otimes}^{\mathcal{A}})_{\nu}}( \left<Q_i, s_{0} \right> \models \Diamond U^{A})$, where $(\mathcal{M}_{\otimes}^{\mathcal{A}})_{\nu}$ denotes the product MC induced by $\mathcal{I} \otimes \mathcal{A}$ corresponding to adversary $\nu$.
\end{fact}\mbox{}

Consequently, computing $\widecheck{\mathcal{P}}_{\mathcal{I}}(Q_i \models \Psi)$ and $\widehat{\mathcal{P}}_{\mathcal{I}}(Q_i \models \Psi)$ amounts to finding the product MCs induced by $\mathcal{I} \otimes \mathcal{A}$ that respectively minimize and maximize the probability of reaching an accepting BSCC from $\left<Q_j,s_0 \right>$. Such reachability problems in IMCs were solved when the destination states are fixed for all induced MCs \cite{chatterjee2008model} \cite{lahijanian2015formal}. 

However, in general, the sets $U^A$ and $U^N$ are not fixed in product IMCs and vary with the assumed values for each transition. Specifically, $U^A$ and $U^N$ are determined by transitions that are either ``on" or ``off", i.e. those whose lower bound is zero and upper bound is non-zero, as seen in the example in Fig. 3: in the product MC $(\mathcal{M}_\otimes^\mathcal{A})_1$ induced by $\mathcal{I} \otimes \mathcal{A}$, $U^A$ is $\lbrace Q_{0}, Q_{1}\rbrace$ while $\lbrace Q_{2} \rbrace$ is non-accepting. In $(\mathcal{M}_\otimes^\mathcal{A})_2$, another product MC induced by $\mathcal{I} \otimes \mathcal{A}$, all states are in $U^N$.

\begin{figure}[t]
\setlength{\belowcaptionskip}{-9pt}
\begin{center}
\includegraphics[scale=0.37]{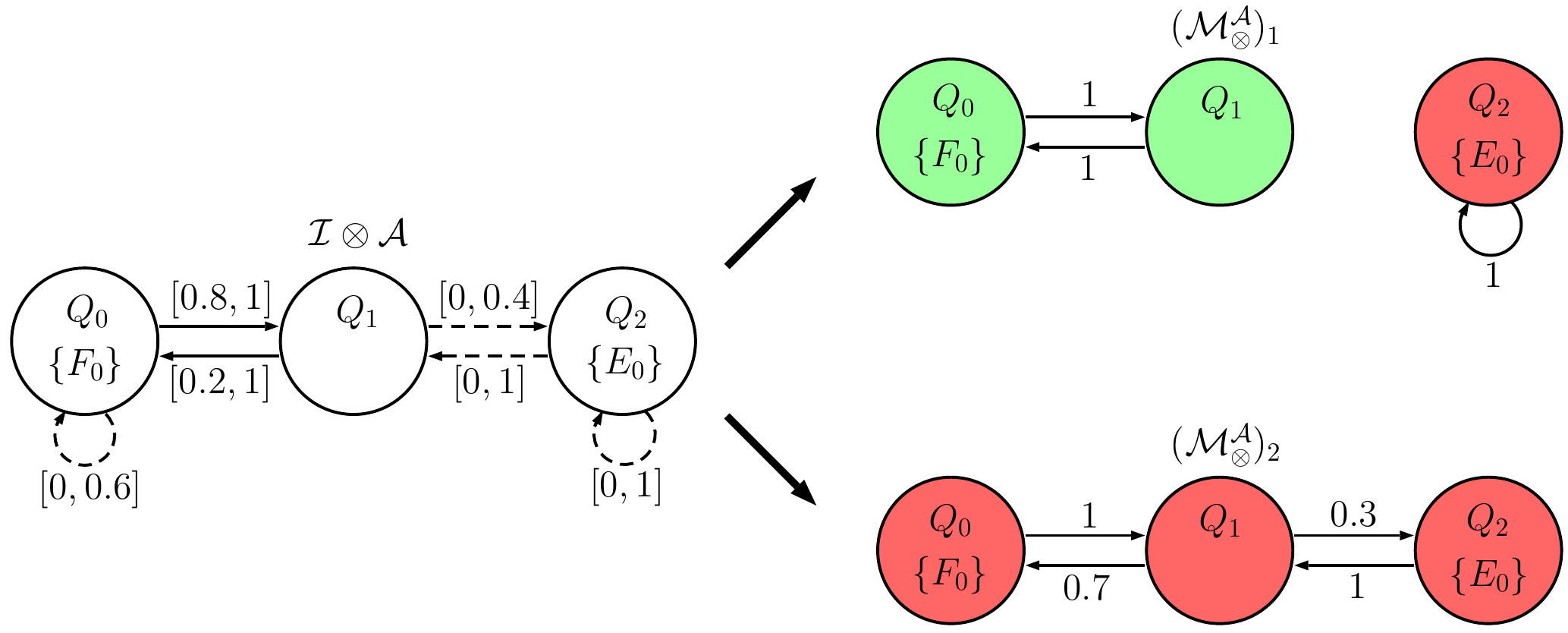}\\
\end{center}
\label{figex1}
\vspace{-0.2cm}
\caption{Two product MCs $(\mathcal{M}_{\otimes}^{\mathcal{A}})_1$ and $(\mathcal{M}_{\otimes}^{\mathcal{A}})_2$ induced by a product IMC $\mathcal{I} \otimes \mathcal{A}$. Accepting BSCCs are shown in green; non-accepting BSCCs are red. In $(\mathcal{M}_{\otimes}^{\mathcal{A}})_1$, $U^{A} = \lbrace Q_{0}, Q_{1} \rbrace$ and $U^{N} = \lbrace Q_{2} \rbrace$; in $(\mathcal{M}_{\otimes}^{\mathcal{A}})_2$, $U^{A} = \emptyset $ and $U^{N} = \lbrace Q_{0}, Q_{1}, Q_{2} \rbrace$.}
\end{figure}

The contribution of this paper to the verification of IMCs for $\omega$-regular properties is twofold. First, we show that a product IMC always induces a largest \textit{Losing Component} and \textit{Winning Component}. These components contain states that reach a BSCC with probability 1. Upper and lower bounds on $\Psi$ are computed by solving a reachability problem for these sets. We further introduce the notion of \textit{Permanent Losing Components} and \textit{Permanent Winning Components} which play a crucial role in the refinement algorithm derived in Section VI. Second, we describe a graph search algorithm to find these components.

\subsection{Computation of Satisfiability Bounds in IMCs}

Previous works highlighted the crucial role of BSCCs in product MCs \cite[Theorem 10.56]{baier2008principles}. As the probability of reaching an accepting BSCC in a product MC determines the probability of satisfying some property in the original abstraction, we now further introduce the notions of winning and losing components. These components include states that may not belong to a BSCC but from which any path is bound to reach a BSCC.\\ 

\begin{definition}[Winning/Losing Component]
\cite{baier2004controller} A \textit{winning (losing) component} $WC \; (LC)$  of a product MC $\mathcal{M}_{\otimes}^{\mathcal{A}}$ is a set of states satisfying $\mathcal{P}_{\mathcal{M}_{\otimes}^{\mathcal{A}}}(WC \models \Diamond U^{A}) = 1 \; ( \; \mathcal{P}_{\mathcal{M}_{\otimes}^{\mathcal{A}}}(LC \models \Diamond U^{N}) = 1 \; )$ , where $U^{A}$ ($U^{N}$) is the set of states belonging to an accepting (non-accepting) BSCC in $\mathcal{M}_{\otimes}^{\mathcal{A}}$.
\end{definition}\mbox{}

It naturally follows that the probability of eventually reaching a BSCC from some initial state is equal to that of reaching a winning or losing component.\\

\begin{corollary}
In any product MC $\mathcal{M}_{\otimes}^{\mathcal{A}}$,
\begin{align}
\mathcal{P}_{\mathcal{M}_{\otimes}^{\mathcal{A}}}( \left< Q_i, s_0 \right> \models \Diamond U^{A}) & =  \mathcal{P}_{\mathcal{M}_{\otimes}^{\mathcal{A}}}( \left< Q_i, s_0 \right> \models \Diamond WC)\\
\mathcal{P}_{\mathcal{M}_{\otimes}^{\mathcal{A}}}( \left< Q_i, s_0 \right> \models \Diamond U^{N}) & =  \mathcal{P}_{\mathcal{M}_{\otimes}^{\mathcal{A}}}( \left< Q_i, s_0\right> \models \Diamond LC) \;.
\end{align}
\end{corollary}\mbox{}

For any initial state in a product IMC, our goal is thus to find induced product MCs that minimize and maximize the probability of reaching a winning component. 


We refer to the technical appendix for all lemmas and proofs leading to the proposed solution. The key observation is that any product IMC induces a \textit{largest winning component} and a \textit{largest losing component}. The largest winning component is the set of states of the product IMC belonging to a winning component for at least one induced product MC, while the largest losing component is the analogous set for losing components. Definitions of \textit{permanent} and \textit{potential} components follow directly from that of largest components.\\

\begin{definition}[Largest Winning/Losing Components]
A state $\left< Q_i, s_j \right> \in Q \times S$ of a product IMC $\mathcal{I} \otimes \mathcal{A}$ is a member of the \textit{Largest Winning (Losing) Component} $(WC)_{L}$ $\big( \; (LC)_{L} \; \big)$ if there exists a product MC induced by $\mathcal{I} \otimes \mathcal{A}$ such that $\left< Q_i, s_j \right>$ is a winning (losing) component.
\end{definition}\mbox{}

\begin{definition}[Permanent Winning/Losing Components]
A state $\left< Q_i, s_j \right> \in Q \times S$ of a product IMC $\mathcal{I} \otimes \mathcal{A}$ is a member of the \textit{Permanent Winning (Losing) Component} $(WC)_{P}$ $\big( \; (LC)_{P} \; \big)$  of $\mathcal{I} \otimes \mathcal{A}$ if $\left< Q_i, s_j \right>$ is a winning (losing) component for all product MCs induced by $\mathcal{I} \otimes \mathcal{A}$.
\end{definition}\mbox{}

\begin{definition}[Potential Winning/Losing Components]
A state $\left< Q_i, s_j \right> \in Q \times S$ of a product IMC $\mathcal{I} \otimes \mathcal{A}$ is a member of the \textit{Potential Winning (Losing) Component} $(WC)_{?}$ $\big( \; (LC)_{?} \; \big)$ of $\mathcal{I} \otimes \mathcal{A}$ if $\left< Q_i, s_j \right> \in (WC)_{L} \setminus (WC)_{P} \;\; \big( \left< Q_i, s_j \right> \in \; (LC)_{L} \setminus (LC)_{P} \; \big)$.
\end{definition}\mbox{}

Note that the sets $(WC)_{?}$ and $(LC)_{?}$ may intersect, and by extension $(WC)_{L}$ and $(LC)_{L}$, while $(WC)_{P}$ and $(LC)_{P}$ are disjoint. An important result established in this paper is that any product IMC induces a set of product MCs where all members of the largest winning component belong to a winning component simultaneously. A product IMC induces an analogous set of product MCs for the largest losing component. We provide proofs in Lemmas 5-7 of the Appendix.

We now state the main result of this section, which establishes that bounds on the probability of satisfying an $\omega$-regular property in an IMC can be computed by solving a reachability maximization problem on a fixed set of states in a product IMC. These sets are the largest components of the product IMC. Furthermore, solving these problems induce sets of best and worst-case product MCs where the probabilities of reaching a winning component are respectively maximized and minimized for all initial states of the product IMC.\\

\begin{theorem}
Let $\mathcal{I}$ be an IMC and $\mathcal{A}$ be a Rabin Automaton corresponding to omega-regular property $\Psi$. Let $(WC)_L$ and $(LC)_L$ be the largest winning and losing components and $(WC)_P$ and $(LC)_P$ be the permanent winning and losing components of the product IMC $\mathcal{I} \otimes \mathcal{A}$. Then for any initial state $Q_i$ of $\mathcal{I}$,
\begin{align}
\label{eq:teo1}
\widecheck{\mathcal{P}}_{\mathcal{I}}(Q_i \models \Psi) & = \;  1 - \mathcal{\widehat{P}}_{\mathcal{I} \otimes \mathcal{A}}( \;  \left<Q_i,s_0 \right> \models \Diamond (LC)_L \; ) \\
\label{eq:teo2} \widehat{\mathcal{P}}_{\mathcal{I}}(Q_i \models \Psi) & = \;  \mathcal{\widehat{P}}_{\mathcal{I} \otimes \mathcal{A}}( \;  \left<Q_i,s_0 \right> \models \Diamond (WC)_L \; ) \; .
\end{align}
\end{theorem}\mbox{}\\
Moreover, there exists a set of induced product MCs $(\mathcal{M}_{\otimes}^{\mathcal{A}})_{\text{worst}}$, where, $\forall \mathcal{M}_{i} \in (\mathcal{M}_{\otimes}^{\mathcal{A}})_{\text{worst}}$, the sets of all losing and winning components of $\mathcal{M}_{i}$ are $(LC)_{L}$ and $(WC)_{P}$ respectively, and, $\forall \left< Q_{i}, s_{0} \right> \in (Q \times S)$, $\mathcal{P}_{\mathcal{M}_{i}}( \left< Q_{i}, s_{0} \right> \models \Diamond (LC)_{L}) = \widehat{\mathcal{P}}_{\mathcal{I} \otimes \mathcal{A}}( \left< Q_{i}, s_{0} \right> \models \Diamond (LC)_{L})$ . Likewise, there exists a set of induced product MCs $(\mathcal{M}_{\otimes}^{\mathcal{A}})_{\text{best}}$, where, $\forall \mathcal{M}_{i} \in (\mathcal{M}_{\otimes}^{\mathcal{A}})_{\text{best}}$, the sets of all losing and winning components of $\mathcal{M}_{i}$ are $(LC)_{P}$ and $(WC)_{L}$ respectively, and, $\forall \left< Q_{i}, s_{0} \right> \in (Q \times S)$, $\mathcal{P}_{\mathcal{M}_{i}}( \left< Q_{i}, s_{0} \right> \models \Diamond (WC)_{L}) = \widehat{\mathcal{P}}_{\mathcal{I} \otimes \mathcal{A}}( \left< Q_{i}, s_{0} \right> \models \Diamond (WC)_{L})$.\mbox{}

\begin{proof}
$\widecheck{\mathcal{P}}_{\mathcal{I}}(Q_i \models \Psi)$ is equivalent to a lower bound on the probability of reaching an accepting BSCC from $\left<Q_i,s_0 \right>$ in $\mathcal{I} \otimes \mathcal{A}$. Equation \eqref{eq:teo1} follows from Lemma 3 and the following reasoning: assume \eqref{eq:teo1} is not true. This implies that there exists an induced product MC where the probability of reaching an non-accepting BSCC from $\left< Q_{i}, s_{0} \right>$ is greater than the highest probability of reaching $(LC)_{L}$, which is a contradiction to Lemma 5 and 9. Next, we denote by $\mathcal{D}$ the set of induced product MCs with set of winning components $(WC)_{P}$ and set of losing components $(LC)_{L}$ constructed in Lemma 10. Lemma 8 and Lemma 9 guarantee that the probability of reaching an accepting BSCC from all $\left<Q_i,s_0 \right>$ is minimized in induced product MCs with the smallest set of winning components and the largest set of losing components respectively. Therefore, $(\mathcal{M}_{\otimes}^{\mathcal{A}})_{\text{worst}} \subseteq \mathcal{D}$. Equation \eqref{eq:teo2} and the existence of $(\mathcal{M}_{\otimes}^{\mathcal{A}})_{\text{best}}$ are proved identically.
\end{proof}\mbox{}

The equalities highlighted by this theorem are central to the elaboration of our verification procedure. We first solve a qualitative problem, which is to find the largest components of the product IMC. This can be achieved via graph search and will be the focus of the next section. Then, we compute upper and lower bound probabilities of reaching these components from all states in the product IMC using existing algorithms found in the literature \cite{chatterjee2008model} \cite{lahijanian2015formal} . By doing so, we construct a best-case product MC $(\mathcal{M}_{\otimes}^{\mathcal{A}})_u \in  (\mathcal{M}_{\otimes}^{\mathcal{A}})_{\text{best}}$ and a worst-case product MC $(\mathcal{M}_{\otimes}^{\mathcal{A}})_l \in  (\mathcal{M}_{\otimes}^{\mathcal{A}})_{\text{worst}}$ which respectively maximizes and minimizes the probability of reaching an accepting BSCC from all initial states. Note that the transition values between states inside the components do not affect the reachability probabilities and do not need to be considered.

\subsection{Winning and Losing Components Search Algorithm}

\begin{algorithm}[t!]
\caption{Find Potential and Permanent BSCCs} 
\begin{algorithmic}[1]
\STATE \textbf{Input}: Product IMC $\mathcal{I} \otimes \mathcal{A}$
\STATE \textbf{Output}: Potential and permanent BSCCs $(U^A)_?$, $(U^A)_P$, $(U^N)_?$, $(U^N)_P$

\STATE \textbf{Initialize}: $(U^A)_?$, $(U^A)_P$, $(U^N)_?$, $(U^N)_P := \emptyset$
\STATE Construct $G = (V,E)$ with a vertex for each state in $\mathcal{I} \otimes \mathcal{A}$ and an edge between states $Q_i$ and $Q_j$ if $\widehat{T}(Q_i, Q_j) > 0$
\STATE Find all SCCs of $G$ and list them in $S$
\FOR {$S_k \in S$}
\STATE $C_0 := \emptyset$,  $i := 0$
\REPEAT
\STATE $R_i := S_k \setminus \cup_{\ell = 0}^{i} C_{\ell}$; \; $Tr_i := V \setminus R_i$; \; $C_{i+1} := At_?(Tr_i, R_i)$; \; $i := i +1$
\UNTIL $C_i = \emptyset$
\IF {$i \not = 1$}
\STATE Find all SCCs of $R_i$ and add them to $S$
\ELSE
\IF{$S_k$ is accepting}
\STATE In $C$, list all states in $S_k$ mapping to some accepting set $F_i$ if no other state in $S_k$ maps to $E_i$. Find all SCCs of $S_k \setminus At_{?}(C, S_k)$ and add them to $S$.
\ELSE
\STATE For all sets $F_i$ to which at least one state in $S_k$ is mapped, set $S'_k = S_k$, list all states mapping to $E_i$ in $C$, find all SCCs of $S'_k \setminus At_{?}(C, S'_k)$ and add them to $S$.
\ENDIF
\IF{$At_{P}(V \setminus S_k, S_k) \not = \emptyset$}
\STATE $(U^A)_? := (U^A)_? \cup \{S_k\}$  or $(U^N)_? := (U^N)_? \cup \{S_k\} $ depending on the acceptance status of $S_k$.
\ELSE
\STATE $ (U^A)_P := (U^A)_P \cup \{S_k\}$  or $(U^N)_P := (U^N)_P \cup \{S_k\}$ depending on the acceptance status $S_k$ and if no other state in $S_k$ belongs to a potential BSCC of the opposite acceptance status. Else, $(U^A)_? := (U^A)_? \cup \{S_k\}$  or $(U^N)_? := (U^N)_? \cup \{S_k\}$.
\ENDIF
\ENDIF
\ENDFOR
\RETURN $(U^A)_?, (U^A)_P , (U^N)_?, (U^N)_P$
\end{algorithmic}
\end{algorithm}

We present a graph-based algorithm for finding $(WC)_P$, $(WC)_L$, $(LC)_P$ and $(LC)_L$ in a product IMC, divided into Algorithm 1 and 2. We define the sets of potential and permanent BSCCs $(U^A)_?$, $(U^A)_P$, $(U^N)_?$ and $(U^N)_P$. Algorithm 1 takes a product IMC as input and returns its potential and permanent BSCCs. Algorithm 2 takes as inputs a product IMC and its permanent and potential BSCCs and outputs $(WC)_P$, $(WC)_?$, $(LC)_P$ and $(LC)_?$. The largest components are the union of the potential and permanent components. 

We employ the following notations: a digraph $G$ is said to be generated by an induced MC $\mathcal{M}_{\otimes}^{\mathcal{A}}$ with transition matrix $T$ and states $Q \times S$ if $G$ has a representative vertex for all states in $\mathcal{M}_{\otimes}^{\mathcal{A}}$, and an edge exists between two such vertices if $T(Q_i, Q_j) > 0$, $Q_i, Q_j \in Q \times S$.  $Reach(S,G)$ denotes the set of vertices in graph $G$ from which there exists a path to the set of vertices $S$; $At_{?}(S,G)$ denotes the set of vertices in $G$ from which there exists a path to $S$ for all graphs $G'$ generated by an induced product MC of $\mathcal{I} \otimes \mathcal{A}$, where $G$ and $G'$ share the same set of vertices; and $At_{P}(S,G)$ denotes the set of vertices in $G$ from which there exists a path to $S$ for at least one graph $G'$ generated by an induced product MC of $\mathcal{I} \otimes \mathcal{A}$. A detailed description of the algorithms is found below.\\

\noindent \textbf{\underline{Algorithm 1}:}\vspace{0.2cm}

\textbf{Line 4}: We first assume all transitions with a non-zero upper bound to be ``on'' and generate a graph $G = (V,E)$ with a vertex for all states and an edge for all transitions in $\mathcal{I} \otimes A$. 

\textbf{Line 5}: Next, we find all strongly connected components (SCC) of $G$ and list them in $S$.

\textbf{Line 6 to 10}: For all SCC $S_k \in S$, we want to determine if there exists an induced MC where $S_k$ is a BSCC. To this end, for all the states $S_k^j$ in $S_k$, we check whether all outgoing transitions to states not in $S_k$ can be turned ``off'' for some induced product MC, that is if the transition lower bounds from $S_k^j$ to states in $Tr_{i} = V \setminus R_i$ are 0 and the sum of the transition upper bounds from $S_k^j$ to states in $S_k$ is greater than 1, which is captured by the use of the function $At_{?}$. Otherwise, $S_k^j$ is said to be \textit{leaky} in all induced product MCs and $S_k^j$ is added to the set $C_i$, which contains all leaky states of $S_{k}$ found at iteration $i$. Note that $R_{0} = S_{k}$ and that all leaky states previously found are removed from $S_k$ at each iteration via variable $R_i$. The loop terminates when all states have been checked and no more leaky states are found, that is $C_{i} = \emptyset$.

\textbf{Line 11 to 13}: If $S_k$ contained leaky states that were previously removed, we compute all SCCs formed by the remaining states in $R_i$ and add them to the list of SCCs of G. If $S_k$ did not contain any leaky state, it is a member of a largest set of BSCCs and the mapping of the states in $S_k$ with respect to the Rabin Pairs decides whether $S_k$ is accepting and $S_k \in (U^A)_L$  or non-accepting and $S_i \in (U^N)_L$. 

\textbf{Line 14 to 15}: If $S_k \in (U^A)_L$, it could still contain potential non-accepting BSCCs, since $(U^A)_{L}$ and $(U^N)_{L}$ may comprise intersecting sets. Treat all states causing $S_k$ to be accepting as leaky (states mapping to some $F_i$ in the Rabin pairs when no states in $S_k$ maps to $E_i$), remove from $S_k$ all states that have a permanent path to the leaky states, compute all SCCs formed by the remaining states and add them to $S$. 

\textbf{Line 16 to 17}:  If $S_k \in (U^N)_L$, potential accepting BSCCs may lie inside $S_k$. For all sets $F_i$ in the Rabin pairs to which at least one state in $S_k$ is mapped, create a ``copy'' $S_k'$ of $S_k$ where all states causing $S_k$ to be non-accepting are considered leaky (the states mapping to $E_i$), remove from $S_k'$ all states that have a permanent path to the leaky states, compute all SCCs formed by the remaining states and add them to S.

\textbf{Line 19 to 22}: We check whether some state in $B$ leaks outside of $B$ for at least one induced MC. If so, the BSCC is not permanent. Otherwise, $B$ is permanent if and only if no BSCC of the opposite acceptance status is found inside of $B$.\vspace{0.3cm}

\begin{algorithm}[t!]
\caption{Find Largest and Permanent Components} 
\begin{algorithmic}[1]
\STATE \textbf{Input}: Product IMC $\mathcal{I} \otimes \mathcal{A}$ and its potential and permanent BSCCs $(U^A)_?, (U^A)_P , (U^N)_?, (U^N)_P$
\STATE \textbf{Output}: Potential and permanent components $(WC)_?$, $(LC)_?$, $(WC)_P$, $(LC)_P$
\STATE \textbf{Initialize}: $(WC)_?$, $(WC)_P$, $(LC)_?$, $(LC)_P := \emptyset$
\STATE Construct $G = (V,E)$ with a vertex for each state in $\mathcal{I} \otimes \mathcal{A}$ and an edge between states $Q_i$ and $Q_j$ if $\widehat{T}(Q_i, Q_j) > 0$
\FOR {$B \in$ $(U^A)_? \cup (U^A)_P \cup (U^N)_? \cup (U^N)_P$}
\STATE $C_0 := \emptyset$, $V_0 := V$, $i := 0$
\REPEAT
\STATE $R_i := Reach(B \cap V_i, V_i)$; $Tr_i := V_i \setminus R_i$; $C_{i+1} := At_{?}(Tr_i, V_i)$ ; $V_{i+1} := V_i \setminus C_{i+1}$ ; $i := i+1$
\UNTIL $C_i = \emptyset$
\STATE $W := V \setminus \cup_{k=1}^{i} C_k$
\IF {$B \in$ $(U^A)_P$ or $B \in$ $(U^N)_P$}
\STATE In $D$, list the states of $V_i$ belonging to a potential BSCC with a different acceptance status from $B$
\STATE $V_P := V_i \setminus D$; $ n := i$
\REPEAT
\STATE $R_i := Reach(B \cap V_i, V_i)$; $Tr_i := V_i \setminus R_i$; $C_{i+1} := At_{P}(Tr_i, V_i)$; $V_{i+1} := V_i \setminus C_{i+1}$ ; $i := i+1$
\UNTIL $C_i = \emptyset$
\STATE $W_P := V_P \setminus \cup_{k=n}^{i} C_k$
\STATE $(WC)_P := (WC)_P \cup \{W_P\}$ or $(LC)_P := (LC)_P \cup \{W_P\}$ depending on $B$
\STATE $(WC)_? := (WC)_? \cup \{W \setminus W_P\}$ or $(LC)_? := (LC)_? \cup \{W \setminus W_P\}$ depending on $B$
\ELSE
\STATE $(WC)_? := (WC)_? \cup \{W\}$ or $(LC)_? := (LC)_? \cup \{W\}$ depending on $B$
\ENDIF
\ENDFOR
\RETURN $(WC)_?$, $(LC)_?$, $(WC)_P$, $(LC)_P$
\end{algorithmic}
\end{algorithm}

 \noindent \textbf{\underline{Algorithm 2}:}\vspace{0.2cm}
 
Inspired by the Classical Algorithm for Buchi MDPs \cite{de1999computing}, we perform a graph search to find permanent and potential winning and losing components for each BSCC. Permanent components only arise from permanent BSCCs, while potential components stem from both potential and permanent BSCCs. 

\textbf{Line 4}: We generate a graph $G = (V,E)$ where transitions with a non-zero upper bound are assumed to be ``on''. 

\textbf{Line 5 to 10} For all BSCCs $B$, we find the set $R_i$ of all states from which there is a path to $B$ in $G$. Other states in $G$ are ``trap states" denoted by $Tr_i$. Then, we iteratively remove the set of states $C_i$ from $R_i$ that ``leak" to $Tr_i$ for all induced MCs, and compute the new set $R_{i+1}$ of states that have a path to $B$ once the leaky states are discarded. The iteration stops when no more leaky states are found, that is  $C_i = \emptyset$. The remaining states belong to $(WC)_L$ or $(LC)_L$ according to $B$.

\textbf{Line 20-21}: If $B$ is a potential BSCC, these states have to belong to one of the potential components --- $(WC)_?$ or $(LC)_?$ --- depending on $B$.

\textbf{Line 11}: If $B$ is a permanent BSCC, we want to check whether the above set of states, denoted by $V_P$, contains members of the permanent components $(WC)_P$ or $(LC)_P$.

\textbf{Line 12}: If $B$ is accepting, remove potential non-accepting component from $V_P$ and treat them as trap state; for a non-accepting $B$, remove potential accepting components instead. 

\textbf{Line 13-20}:  Repeat the same procedure as in Algorithm 1, except that leaky states are now those which have a path to the trap states in at least one induced MC. The remaining states are permanent components of the same acceptance status as $B$.\\

To summarize, it is known that verification against temporal logic specifications in discrete-time MCs can be accomplished by solving a reachability problem on a product MC constructed from a Rabin automaton corresponding to the specification to be verified. The heart of this approach relies on analyzing winning and losing components of the product MC. These ideas do not directly extend to IMCs because BSCCs are not uniquely determined in this case; this is because some transitions can have a lower transition bound equal to 0 but an upper transition bound that is non-zero. Instead, we introduced the concepts of largest winning and losing components. In Theorem 2, we show that upper and lower bounds on the probability of satisfaction are obtained from these components. Algorithms 1 and 2 provide means for computing these components. Note that the proposed algorithm allows to perform verification of IMCs without constructing an exponentially large Markov decision process, as done in \cite{chatterjee2008model}.

\begin{algorithm}[t!]
\caption{State-Space Refinement Scoring Procedure} 
\begin{algorithmic}[1]
\STATE \textbf{Input}: Worst-case product MC $(\mathcal{M}_{\otimes}^{\mathcal{A}})_{l}$ and best-case product MC $(\mathcal{M}_{\otimes}^{\mathcal{A}})_{u}$ induced by the product $\mathcal{I} \otimes \mathcal{A}$
\STATE \textbf{Output}: Scores $\sigma = \left [ \sigma_{0}, \ldots, \sigma_{N} \right]$ for all states in $\mathcal{I}$
\STATE \textbf{Initialize}: $\sigma_i = 0$, with $\sigma_i$ the score of the $i$-th state of $\mathcal{I}$, $p_{stop} \in (0, 1)$ user-defined probability threshold 
\FOR {$Q_{\ell} \in Q_{\phi}^{?}$}
\STATE $\pi := q_0 := \left<Q_{\ell}, s_0 \right>$ in $(\mathcal{M}_{\otimes}^{\mathcal{A}})_{u}$
\REPEAT
\IF {$\mathcal{P}(\pi) < p_{stop}$ or $Exp(\pi) = R(\pi)$}
\STATE $\pi := \pi^{-}$
\ELSE
\IF{$Exp(\pi) \not = \emptyset$}
\STATE $q_i \rightarrow Exp(\pi)$, where $q_i$ is any state in $R(\pi) \setminus Exp(\pi)$; $\pi := \pi^{+}(q_i)$
\ELSE
\STATE $Exp(\pi) := \cup_{i} \pi_i$
\IF{$Last(\pi) \in (WC)_{?} \cup (LC)_{?}$}
\STATE $\sigma_j := \sigma_j + \mathcal{P}(\pi)(p_{max} - p_{min})$ for all states $\left<Q_j, s_i \right>$ in the potential BSCC of $ Last(\pi)$ with an outgoing transition which can be either zero or non-zero,  $p_{max}$ and $p_{min}$ are the probabilities of reaching an accepting BSCC  from $ Last(\pi)$ in $(\mathcal{M}_{\otimes}^{\mathcal{A}})_{u}$ and $(\mathcal{M}_{\otimes}^{\mathcal{A}})_{l}$ respectively; $\pi := \pi^{-}$
\ELSIF{$Last(\pi) \in (WC)_{P} \cup (LC)_{P}$}
\STATE $\pi := \pi^{-}$
\ELSE
\STATE $\sigma_j := \sigma_j + \mathcal{P}(\pi)(p_{max} - p_{min})$, where $j$ corresponds to $\left<Q_j, s_i \right> := Last(\pi)$, $p_{max}$ and $p_{min}$ are as in line 15;
\STATE $\pi :=\pi^{+}(q_i)$ where $q_i$ is any state in $R(\pi)$
\ENDIF
\ENDIF
\ENDIF
\UNTIL $\pi = \emptyset$
\ENDFOR
\RETURN $\sigma$
\end{algorithmic}
\end{algorithm}

\section{State-Space Refinement}

Given a partition $P$ of the domain $D$ and a specification $\phi$ as in \eqref{eq2}, the verification procedure derived in Section V assigns each discrete state of $P$ to one of the sets $Q_{\phi}^{yes}$, $Q_{\phi}^{no}$ or  $Q_{\phi}^{?}$. One aims to find a partition $P$ that yields a low volume of undecided states in $Q_{\phi}^{?}$. To this end, we suggest a specification-guided iterative method. Specifically, we first generate a rough partition $P$ of $D$ and successively refine $P$ into finer partitions by targeting the best candidate states for reducing the uncertainty in the abstraction with respect to $\phi$. These states are chosen after comparing the behavior of the system in the best and worst-case scenarios computed during verification. The procedure stops when a user-defined criterion is reached. Here, we terminate when the fractional volume of uncertain states is less than a threshold $V_{stop} $ $\in [0,1]$.

We seek to analyze the behavior of accepting paths in the best and worst-case product MCs $(\mathcal{M}_{\otimes}^{\mathcal{A}})_{u}$ and $(\mathcal{M}_{\otimes}^{\mathcal{A}})_{l}$ obtained at the time of verification and illustrated in Fig. 4. In particular, for every undecided state $Q_{j}$ in $Q_{\phi}^{?}$, we look at all paths starting from $ \left<Q_j, s_0 \right>$ in $(\mathcal{M}_{\otimes}^{\mathcal{A}})_{u}$ and assign a score to the states encountered along them depending on how these states behave in $(\mathcal{M}_{\otimes}^{\mathcal{A}})_{l}$.  We inspect a path until it reaches a state that belongs to either $(WC)_{L}$ or $(LC)_{L}$, or when its probability of occurrence in $(\mathcal{M}_{\otimes}^{\mathcal{A}})_{u}$ falls below a threshold $p_{stop}$. States with high scores are targeted for refinement.

We introduce some notation: for a finite path $\pi = q_0 q_1 \ldots q_k$ in $(\mathcal{M}_{\otimes}^{\mathcal{A}})_{u}$, $Last(\pi)$ denotes the last state $q_k$ of $\pi$; $\pi_{i}$ denotes the $i$-th state of $\pi$; $\mathcal{P}(\pi) = T(q_0, q_1) \cdot T(q_1, q_2) \cdot \ldots \cdot T(q_{k-1}, q_{k})$, $\mathcal{P}(q_0) = 1$, is the probability of path $\pi$ in $(\mathcal{M}_{\otimes}^{\mathcal{A}})_{u}$; $R(\pi)$ is the set of states that are one-step reachable from $Last(\pi)$ in $(\mathcal{M}_{\otimes}^{\mathcal{A}})_{u}$; $Exp(\pi)$ denotes all continuations of $\pi$ from $R(\pi)$ that have been explored and is initialized to the empty set for all $\pi$; $\pi^{-}$ is the path obtained by removing the last state of $\pi$ and $\pi^{+}(q_i)$ is the path with $q_i$ appended to $\pi$. $V_?$ is the fractional volume of uncertain states and is equal to the sum of the volume of all states in $Q_{\phi}^{?}$ divided by the volume of the domain $D$. Our procedure is as follows: \vspace{0.2cm}

\begin{enumerate}[leftmargin = 14pt]
\item Compute a refinement score for all states in $\mathcal{I}$ according to Algorithm 3, which is described below:\vspace{0.2cm}

\hspace{2mm} \textbf{Line 6 to 24}: This loop terminates when $\pi = \emptyset$, that is, when all paths starting from $\left<Q_j, s_0 \right>$ have been explored.

\hspace{1.5mm} \textbf{Line 7 to 8}: If $\mathcal{P}(\pi) < p_{stop}$ or $Exp(\pi) = R(\pi)$, the probability of the path is below the pre-defined exploration threshold or all continuations of $\pi$ have been explored. Thus, we return to the previous state in the path.

\hspace{2mm} \textbf{Line 10 to 13}: lf $Exp(\pi) \not = \emptyset$, add $q_i$ to $Exp(\pi)$, where $q_i$ is some unexplored state in $R(\pi)$ and extend the path to $q_i$. Else, $\pi$ is a path fragment which has not been explored yet. Add all states in $\pi$ to $Exp(\pi)$ to avoid loops.

\hspace{2mm} \textbf{Line 14 to 15}: If $Last(\pi) \in (WC)_{?}$ or $Last(\pi) \in (LC)_{?}$, the path reached a state in a potential component. We want to target the states which can either confirm or refute that $Last(\pi)$ belongs to such a component. These states are the ones inside the potential BSCCs that $Last(\pi)$ belongs to (or makes a transition to with probability 1)  that have outgoing transitions which can be either ``on'' or ``off", as depicted in Fig. 5. A potential "certainty gain" is added to the score of all such states and the path is returned to its previous state. If $Last(\pi)$ belongs to both $(WC)_{?}$ and $(LC)_{?}$, then the scoring scheme is applied to all intersecting potential BSCCs related to $Last(\pi)$. This heuristical gain quantifies a potential reduction in the width of the satisfaction interval of $Q_{\ell}$ in the scenario that the refinement of the considered states provides perfect information, i.e., the probability of reaching an accepting BSCC from $Last(\pi)$ becomes a fixed number.
  
\hspace{2mm} \textbf{Line 16 to 17}: If $Last(\pi) \in (WC)_{P}$ or $Last(\pi) \in (LC)_{P}$, the path reached a region of the state-space that does not require refinement as it belongs to a permanent component. The path returns to its previous state. 

\hspace{2mm} \textbf{Line 18 to 20}: Else, $Last(\pi)$ does not belong to a winning or losing component for any refinement of the product IMC. The potential ``certainty gain'' one can hope for by refining $\left<Q_j, s_i \right> = Last(\pi)$ is added to the score of $Q_j$.  The path is continued to an unexplored state.\vspace{0.2cm}

\item Refine the states in $P_k$ with scores above a user-defined threshold to generate $P_{k+1}$. \vspace{0.2cm}
\item Generate an IMC abstraction of the system with respect to $P_{k+1}$, perform model-checking and compute $V_{?}$.\vspace{0.2cm}
\item If $V_{?} > V_{stop}$, return to step 1. Else, terminate.\vspace{0.2cm}
\end{enumerate}

\begin{figure}[t]
\setlength{\belowcaptionskip}{-10pt}
\begin{center}
\includegraphics[scale=0.34]{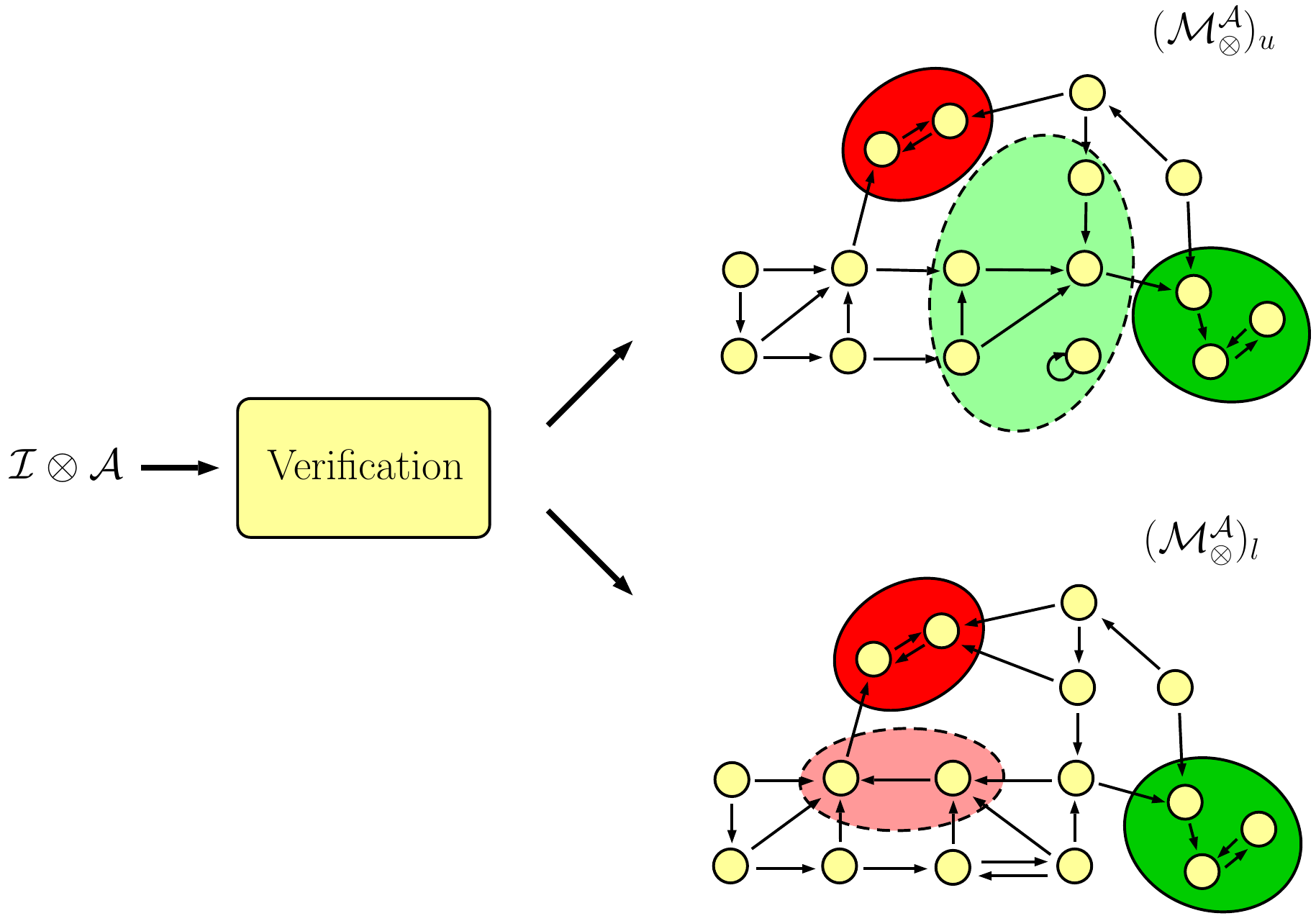}\\
\end{center}
\vspace{-0.3cm}
\caption{Our IMC verification algorithm generates a best and worst-case product MC $(\mathcal{M}_{\otimes}^ {\mathcal{A}})_u$ and $(\mathcal{M}_{\otimes}^{ \mathcal{A}})_l$. Winning and losing components are respectively in red and green; permanent and potential components are circled in bold and dotted lines respectively. Comparing the behavior of the paths in the two scenarios is the basis of our refinement algorithm.}
\end{figure}

\begin{figure}[t]
\setlength{\belowcaptionskip}{-10pt}
\begin{center}
\includegraphics[scale=0.37]{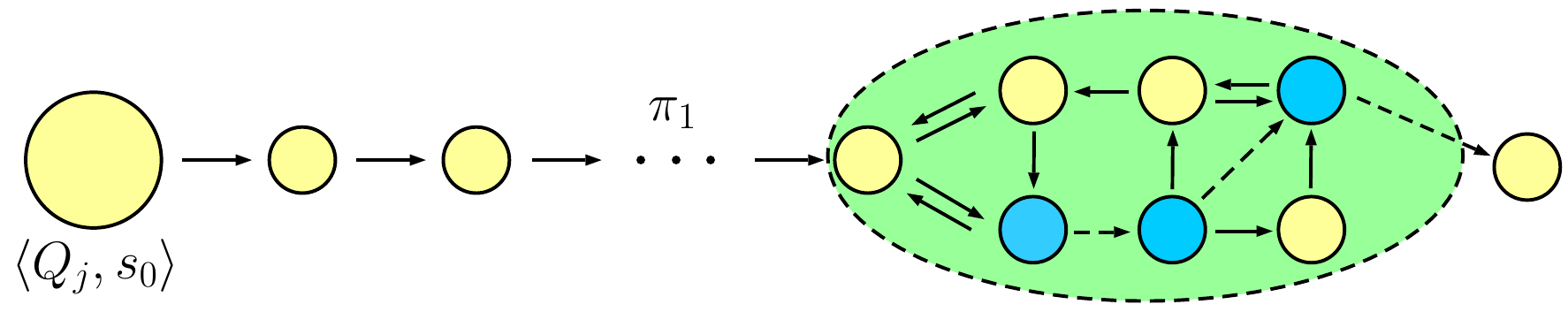}\\
\end{center}
\vspace{-0.2cm}
\caption{For all undecided states $Q_j$ of the IMC abstraction $\mathcal{I}$, we inspect all paths starting from $\left<Q_j,s_0 \right>$ in $(\mathcal{M}_{\otimes}^{\mathcal{A}})_u$ to determine which states to refine. Above is an example of a path $\pi_1$. A score is assigned to all states along $\pi_1$ as detailed in Section VI. In particular, if $\pi_1$ reaches a member of a potential BSCC, a score is assigned to the states which could possibly destroy the BSCC under refinement. These states are shown in blue and have outgoing transitions with lower bound 0.}
\end{figure}

\noindent It is not difficult to construct examples demonstrating that the volume of uncertain states $V_{?}$ need not decrease monotonically at each step of the refinement algorithm using our abstraction technique. This is because, when a parent state is refined to two children states, the sum of the upper transition bounds for the children states may be greater than the upper transition bound of the original parent state. Nevertheless, when $\mathcal{F}$ is continuous, the size of the reachable sets, and consequently the error in the transitions, approaches zero as the grid size decreases. Thus, in the limit, the volume of uncertain states $V_{?}$ decreases to zero.

A common refinement approach consists in splitting the chosen states in the partition in half along their greatest dimension. As the scoring procedure in Algorithm 3 may select the entire state-space of IMC $\mathcal{I}$ for refinement, the worst-case growth of the size of the product IMC $\mathcal{I} \otimes \mathcal{A}$ is exponential and scales in $\mathcal{O}(|S| \cdot 2^{|Q|})$, where $|S|$ and $|Q|$ are the number of states of automaton $\mathcal{A}$ and IMC $\mathcal{I}$ respectively.

However, because this path-based scoring procedure aims to target states which are most likely to reduce the volume of undecided states in the partition with respect to the specification under consideration, our refinement algorithm tends to focus on specific regions of the state-space, as shown in the next section, with the effect of mitigating state explosion. In particular, we were able to achieve lower volumes $V_{stop}$ before the number of states in the partition became prohibitive compared to our naive refinement method in \cite{dutreix2018} which systematically refined all uncertain states.\mbox{}

All algorithms discussed in this paper are implemented in a python package available at \url{https://github.com/gtfactslab/TACStochasticVerification}.

\section{Case Study}

We now apply our verification and refinement procedure in a case study. We consider a nonlinear, monotone bistable switch system with additive disturbance and governing equations
\begin{equation}
  \begin{aligned}
x_{1}[k+1] & = x_1[k]  + ( \; -a x_{1}[k] + x_{2}[k] \; ) \cdot \Delta T + w_1\\
x_{2}[k+1] & = x_2[k]  + \Big(\; \frac{(x_{1}[k])^{2}}{(x_{1}[k])^{2} + 1} - b x_{2}[k] \; \Big) \cdot \Delta T + w_2 \;\; ,
 \end{aligned}
\label{eq:90}
\end{equation}

\noindent where we assume $w_1$ and $w_2$ to be independent truncated Gaussian random variables sampled at each time step. $w_1\sim \mathcal{N}(\mu = -0.3 ; \sigma^2 = 0.1)$ and is truncated on $[-0.4, -0.2]$; $w_2$ is identical. To keep the system self-contained in $D$, we assume that any time the disturbance would push the trajectory outside of $D$, it is actually maintained on the boundary of $D$. This assumption reflects the behavior of systems with bounded capacity where the state variables are restricted to some intervals. We choose $a = 1.3$, $b =0.25$ and $\Delta T =0.05$. The deterministic piece of the system has two stable equilibria at $(0,0)$ and $(2.71, 3.52)$ and one unstable equilibrium. We seek to verify \eqref{eq:90} on a domain $D$, with initial rectangular partition $P$ depicted in Fig. 6 (Top) and Fig. 7 (Top), against the probabilistic LTL specifications 
\begin{align}
\phi_1 & = \mathcal{P}_{\geq 0.80}[\square((\neg A \wedge  \bigcirc A) \rightarrow (\bigcirc \bigcirc A \wedge \bigcirc \bigcirc \bigcirc A))]\\
\phi_2 & = \mathcal{P}_{\leq 0.90}[( \lozenge \square A \rightarrow \lozenge B) \wedge (\lozenge C \rightarrow \square \neg B)] \ .
\end{align}
\noindent Specification $\phi_1$ translates in natural language to ``trajectories that have more than a 80\% chance of remaining in an $A$ state for at least 2 more time steps when entering an $A$ state''. Specification $\phi_2$ translates to ``trajectories that have less than a 90\% chance of reaching a $B$ state if it eventually always remain in $A$, and of always staying outside of $B$ if it reaches a $C$ state''. Their Rabin automaton representations contain 5 and 7 states respectively. We perform verification with stopping criterion $V_{stop} = 0.13$ for $\phi_1$ and $V_{stop} = 0.1$ for $\phi_2$. To construct IMC abstractions of this system, we use the technique shown in Section IV. Graph search is based on Section V-B and we compute reachability bounds applying the algorithm in \cite{lahijanian2015formal}. Upon verification, we select states with an uncertainty score as defined in Section VI that is greater than $10\%$ of the highest score for refinement. Selected states are split into two rectangles along their largest dimension to keep the new partition rectangular. The procedure was conducted on a 3.3 GHz Intel Core i7 with 8 GB of memory using Python.

\begin{figure}[t]
\setlength{\belowcaptionskip}{-10pt}
\begin{center}
\includegraphics[scale=0.165]{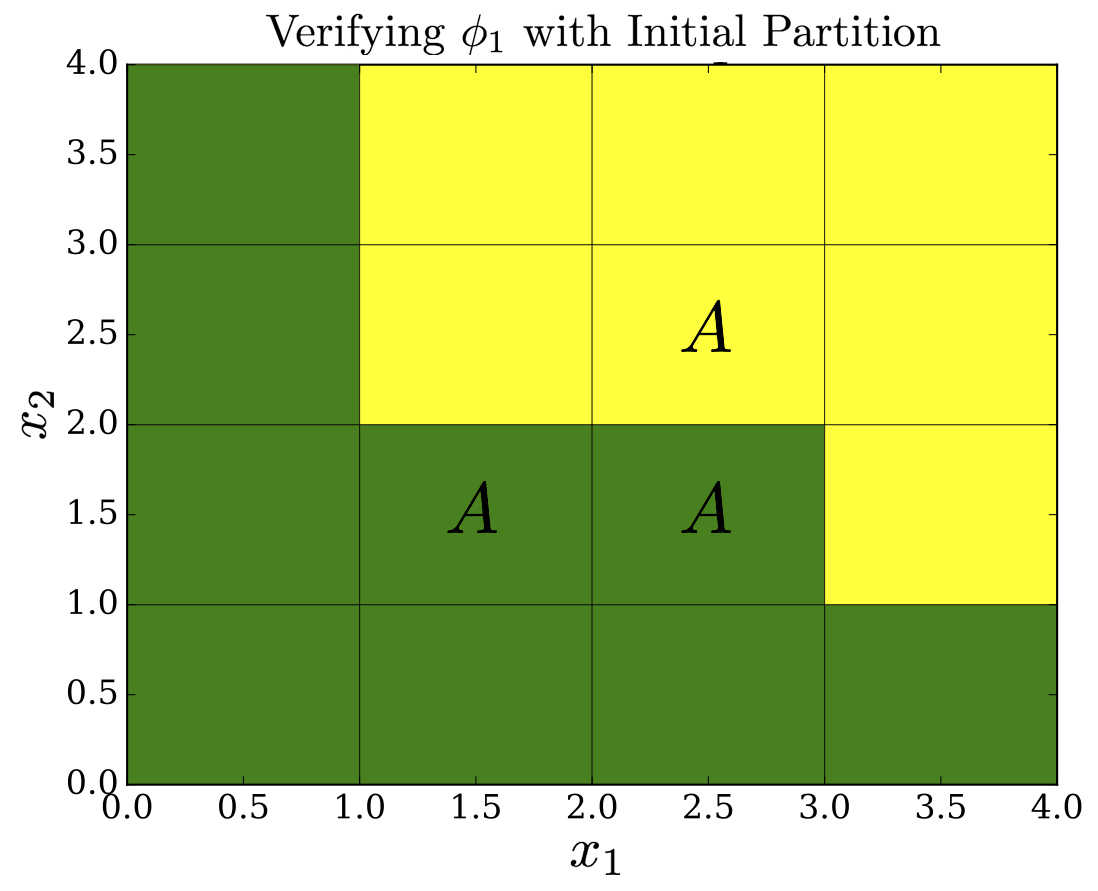}
\includegraphics[scale=0.345]{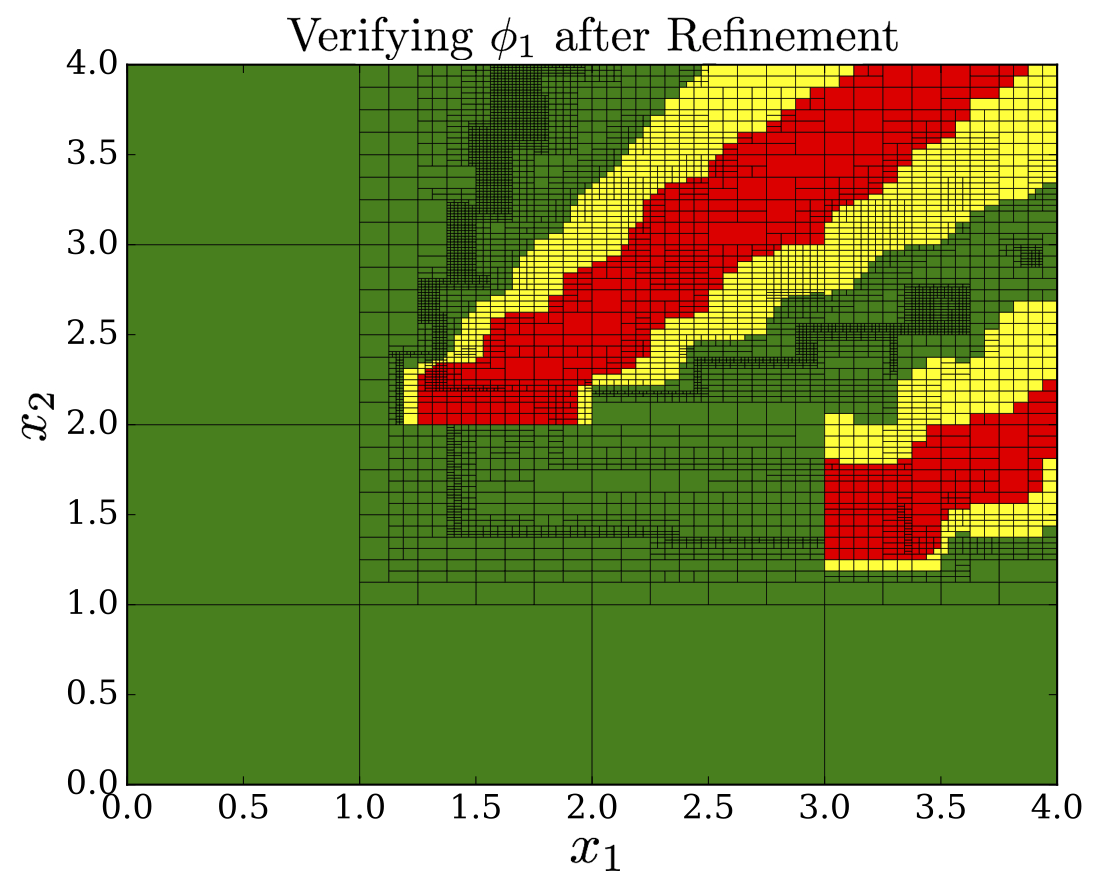}
\end{center}
\vspace{-0.5cm}
\caption{Initial verification of a partition of domain $D$ for specification $\phi_1$ (Top), and verification of final partition (Bottom). States satisfying $\phi_1$ are in green, states violating $\phi_1$ are in red, undecided states are yellow.}
\end{figure}

\begin{figure}[t]
\setlength{\belowcaptionskip}{-10pt}
\begin{center}
\includegraphics[scale=0.165]{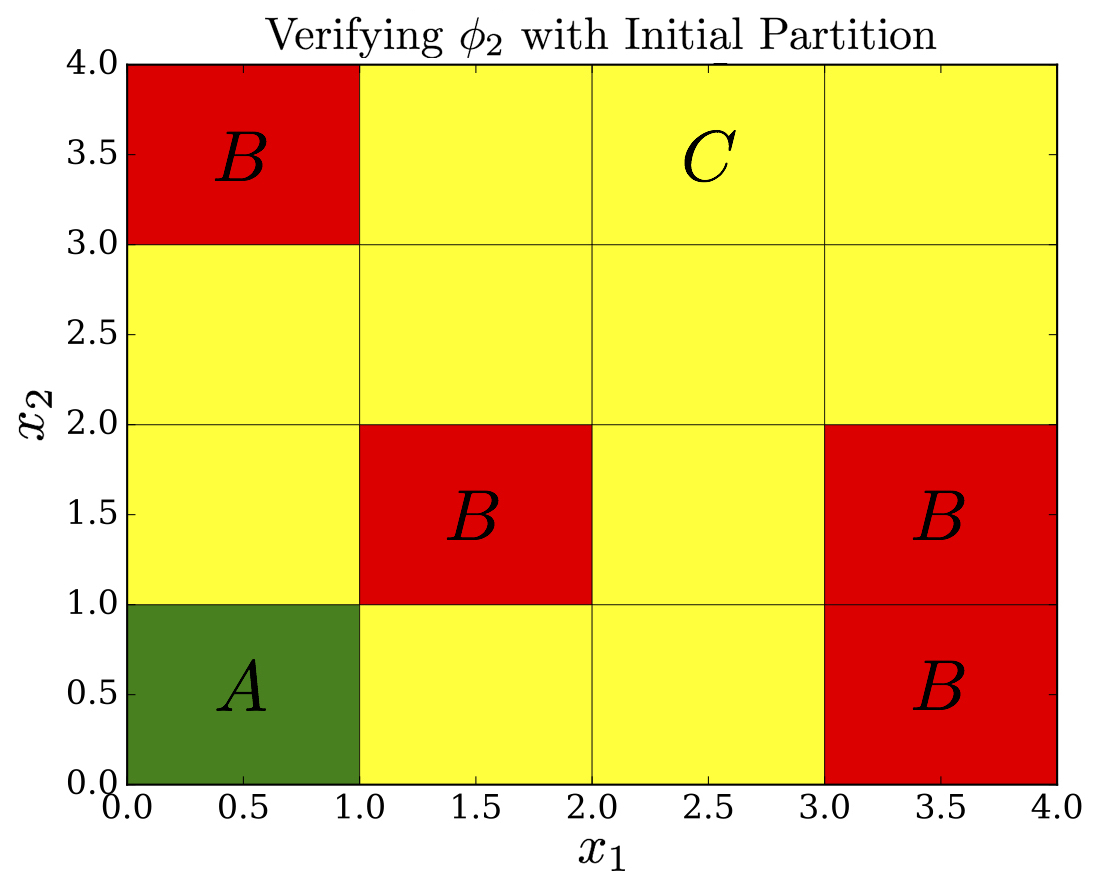}
\includegraphics[scale=0.345]{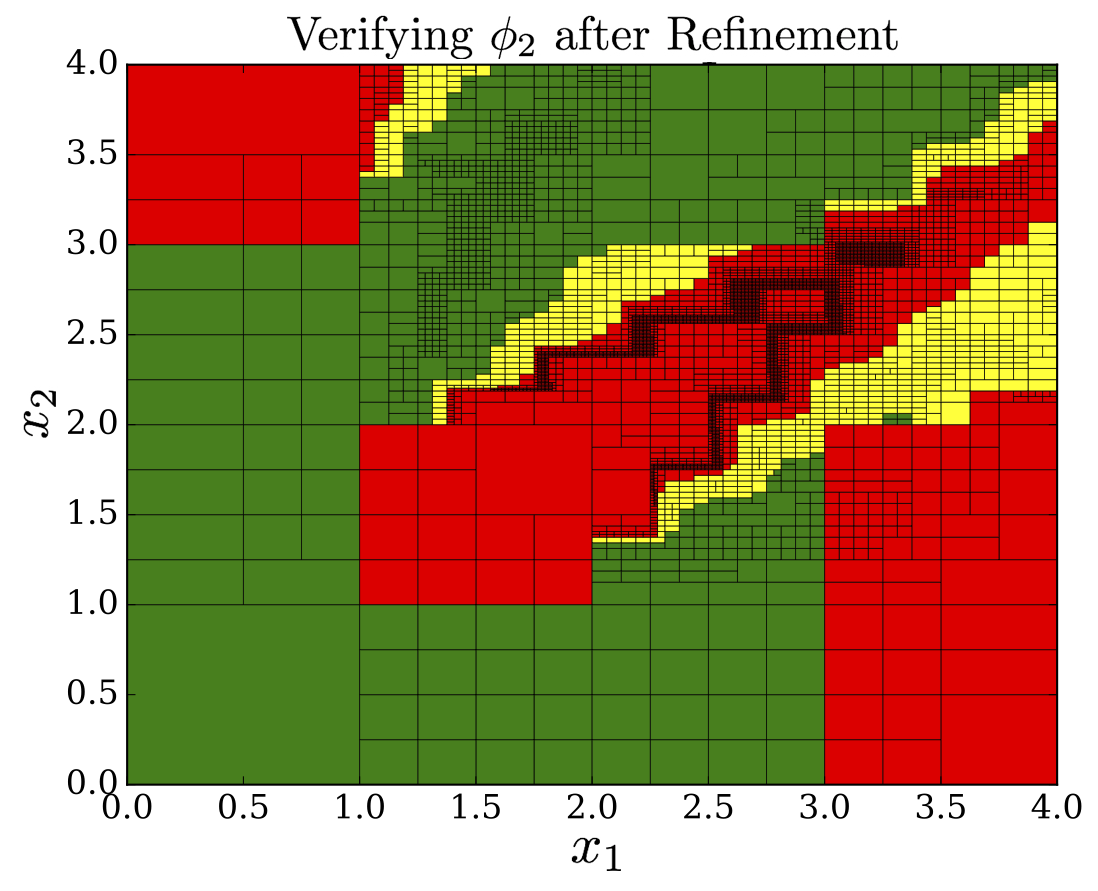}
\vspace{-0.3cm}
\caption{Initial verification of a partition of domain $D$ for specification $\phi_2$ (Top), and verification of final partition (Bottom). States satisfying $\phi_2$ are in green, states violating $\phi_2$ are in red, undecided states are yellow.}
\end{center}
\end{figure}

For $\phi_1$, the refinement algorithm produced 3531 states and terminated in 1h56min after 12 refinement steps. For $\phi_2$, it generated 4845 states and terminated in 3h15min after 13 steps. The final partitions are shown in Fig. 6 (Bottom) and Fig. 7 (Bottom). Our new method outperforms the algorithm we propose in \cite{dutreix2018} which refines all undecided states at each refinement step: for instance, for $\phi_1$, \cite{dutreix2018} achieves $V_? = 0.2137$ in 2 hours 58 min and 11 steps. 
Our algorithm non-uniformly refined the initial partition across the state-space. In the first example, the boundary between regions which can and cannot reach an $A$ state are heavily targeted, as well as boundaries between regions which could keep the system in an $A$ state for one and two time steps. In the second example, the edges of a region leading to $A$ via $B$ are refined the most, as this region is critical with respect to $\phi_2$. Although these two examples share the same dynamics, our algorithm generates very different partitions depending on the specification. Therefore, specification-free gridding approaches are likely to perform conservatively for these examples.

Because this procedure produces different partitions for different temporal objectives, our algorithm is suited for the verification of systems against well-identified specifications which are known a priori, while FAUST${}^2$, which provides error guarantees with respect to entire classes of temporal logic formulas, better accommodates situations requiring an analysis against numerous distinct specifications. Furthermore, unlike our work, FAUST${}^2$ only permits verification for specifications with bounded-time temporal operators \cite{cauchi2019stochy}.

On the other hand, the toolbox StocHy \cite{cauchi2019stochy} allows for the verification of certain unbounded-time operators through the use of IMC abstractions. However, because its verification algorithm is based on \cite{lahijanian2015formal}, it solely accommodates specification belonging to a fragment of LTL and not all $\omega$-regular 
properties. In addition, StocHy employs the IMC abstraction technique presented in \cite{cauchi2019efficiency} which applies only to affine-in-disturbance linear systems, while the abstraction method shown in this work and used in the case study involves the wide class of stochastic mixed monotone dynamics.

Finally, though similar in flavor, the objective of our verification technique is not identical to the one of FAUST${}^2$ and StocHy. The latter aim to create partitions engendering a user-defined abstraction error, whereas the goal of this work is to obtain a small volume of uncertain state with respect to a probabilistic $\omega$-regular specification. Note that low-error abstractions can generate a high volume of uncertain states while high-error abstractions can produce a low volume of uncertain states, depending on the threshold $p_{sat}$. 



\section{Conclusion}

In this paper, we described an algorithm for performing verification against $\omega$-regular properties in continuous state stochastic systems. The proposed approach relies on computing a finite state abstraction of the original system in the form of an IMC and can accommodate classes of specifications not previously covered in the literature. We have developed an efficient procedure for computing the IMC of a mixed monotone system with affine disturbance over rectangular partitions. Furthermore, we presented a specification-guided strategy for refining a finite partition of the continuous domain until a precision threshold has been met. Our technique resolves qualitative issues highlighted in the literature for abstraction-based methods by targeting states that are likely to confirm or destroy winning and losing components in a product IMC. We showed the practicality of this approach in two examples.
%


%

\appendices
\section{Proofs of Section IV}
\subsection{Proof of Proposition \ref{prop:1}}
\begin{proof}
By Proposition \ref{prop:reach}, we observe
\begin{align}
\label{eq:8}
  \{\F(x):x\in Q_1\}\subseteq \{z: g(a^1,b^1)\leq z\leq g(b^1,a^1)\}.\quad
\end{align}
To prove \eqref{eq:2}, we have
\begin{align}
\nonumber&  \min_{x\in Q_1} Pr(\mathcal{F}(x)+w\in Q_2)\\
\label{eq:7}&\geq \min_{z: g(a^1,b^1)\leq z\leq g(b^1,a^1)}Pr(z+w\in Q_2)\\
\label{eq:7-2}&=\min_{z: g(a^1,b^1)\leq z\leq g(b^1,a^1)}\prod_{i=1}^n Pr(z_i+w_i\in [a^2_i,b^2_i])\\
\label{eq:7-3}&=\prod_{i=1}^n\min_{z_i: g_i(a^1,b^1)\leq z_i\leq g_i(b^1,a^1)} Pr(z_i+w_i\in [a^2_i,b^2_i])\qquad
\end{align}
where \eqref{eq:7} follows from \eqref{eq:8}, \eqref{eq:7-2} follows from the mutual independence of all components of $w$ in Assumption \ref{assum:indep}, and \eqref{eq:7-3} holds because $g(a^1,b^1)\leq z\leq g(b^1,a^1)$ if and only if $g_i(a^1,b^1)\leq z_i\leq g_i(b^1,a^1)$ for all $i=1,\ldots, n$. Then \eqref{eq:2} holds since $Pr(z_i+w_i\in [a^2_i,b^2_i])=\int_{a_i^2}^{b_i^2}f_{w_i}(x-z_i) dx$. By a symmetric argument replacing $\min$ with $\max$, \eqref{eq:5} holds.
\end{proof}

\subsection{Proof Lemma \ref{lem:main}}
\begin{proof}
For $s\in\mathbb{R}$, define $H(s)=\int_{a}^{b} f_{\omega}(x-s) \; dx$. We claim
\begin{align}
  H(s_{max})=\max_{s\in\mathbb{R}} H(s),
\end{align}
and that, for all $s_1,s_2\in\mathbb{R}$ such that $|s_{max}-s_1|\geq |s_{max}-s_2|$, it holds that $H(s_1)\leq H(s_2)$, that is, $H(s)$ monotonically decreases as $|s_{max}-s|$ increases. Assuming the claim to be true, it follows that $\max_{s\in[r_1,r_2]} H(s)=H(s^r_{max})$ and $\min_{s\in[r_1,r_2]} H(s)=H(s^r_{min})$, \emph{i.e.}, \eqref{eq:6} and \eqref{eq:10}, completing the proof. To prove the claim, we have, for all $s\in\mathbb{R}$,
\begin{align}
\nonumber&  H(s_{max})-H(s)\\
&=\int_{a}^{b} f_{\omega}(x-s_{max}) \; dx - \int_{a}^{b} f_{\omega}(x-s) \; dx\\
&=\int_{a-s_{max}}^{b-s_{max}} f_{\omega}(x) \; dx - \int_{a-s}^{b-s} f_{\omega}(x) \; dx\\
&=\int_{a-s_{max}}^{a-s} f_{\omega}(x) \; dx - \int_{b-s_{max}}^{b-s} f_{\omega}(x) \; dx\\
\label{eq:12}&=\int_{0}^{s_{max}-s} \bigg[{\textstyle f_\omega(x+\frac{a-b}{2}-c)-f_\omega(x+\frac{b-a}{2}-c)}\bigg] \; dx. \quad \quad
\end{align}
Moreover, because $f_\omega$ is symmetric and unimodal with mode $c$, $f_\omega(x+\frac{a-b}{2}-c)-f_\omega(x+\frac{b-a}{2}-c)$
is an odd function of $x$ and is negative for $x>0$ and positive for $x<0$. Therefore, the integral in \eqref{eq:12} is nonnegative and monotonically decreases as $|s_{max}-s|$ increases, thus proving the claim.
\end{proof}





\subsection{Proof of Theorem \ref{thm:1}}
\begin{proof}
For all $Q_j,Q_\ell\in P$ and $i=1,\ldots,n$, 
  \begin{align}
\nonumber \min_{\substack{z_i\in[ g_i(a^j,b^j), g_i(b^j,a^j)]}}&\int_{a_i^\ell}^{b_i^\ell}f_{w_i}(x-z_i) dx       \\
    \label{eq:34}&= \int_{a^{\ell}_{i}}^{b^{\ell}_{i}} f_{w_{i}}(x_{i}-s_{i,min}^{j \rightarrow \ell})\; dx_{i},\\
\nonumber \max_{\substack{z_i\in[ g_i(a^j,b^j), g_i(b^j,a^j)]}}&\int_{a_i^\ell}^{b_i^\ell}f_{w_i}(x-z_i) dx\\
       &= \int_{a^{\ell}_{i}}^{b^{\ell}_{i}} f_{w_{i}}(x_{i}-s_{i,max}^{j \rightarrow \ell})\; dx_{i},    
  \end{align}
by Lemma \ref{lem:main}. Then, by Proposition \ref{prop:1}, 
\begin{align}
  \label{eq:35}
  \min_{x\in Q_j} Pr(\mathcal{F}(x)+w\in Q_\ell)\geq \prod_{i=1}^n \int_{a^{\ell}_{i}}^{b^{\ell}_{i}} f_{w_{i}}(x_{i}-s_{i,min}^{j \rightarrow \ell})\; dx_{i}\\
  \label{eq:35-2}  \max_{x\in Q_j} Pr(\mathcal{F}(x)+w\in Q_\ell)\leq \prod_{i=1}^n \int_{a^{\ell}_{i}}^{b^{\ell}_{i}} f_{w_{i}}(x_{i}-s_{i,max}^{j \rightarrow \ell})\; dx_{i},
\end{align}
so that \eqref{eq:15}--\eqref{eq:16} implies \eqref{eq:3}--\eqref{eq:3-2}. Furthermore, \eqref{eq:35}--\eqref{eq:35-2} implies $\widecheck{T}(Q_j,Q_\ell)\leq \widehat{T}(Q_j,Q_\ell)$ and \eqref{eq:3}--\eqref{eq:3-2} implies \eqref{eq:36} so that $\mathcal{I}=(P, \widecheck{T},\widehat{T})$ is a valid IMC, concluding the proof.
\end{proof}

\section{Lemmas and Proofs of Section V}


\begin{lemma}
\cite{baier2008principles} For any infinite sequence of states $\pi = q_{0}q_{1}q_{2}\ldots$ in a Markov Chain, there exists an index $i \geq 0$ such that $q_i$ belongs to a BSCC.
\end{lemma}

\begin{corollary}
For any initial state $ \left< Q_i, s_0 \right>$ in an induced product MC $\mathcal{M}_{\otimes}^{\mathcal{A}}$,
\begin{align}
\mathcal{P}_{\mathcal{M}_{\otimes}^{\mathcal{A}}}( \left<Q_i,s_0 \right> \models \Diamond U^A) + \mathcal{P}_{\mathcal{M}_{\otimes}^{\mathcal{A}}}(\left< Q_i,s_0 \right> \models \Diamond U^N) = 1 \ .
\end{align}
\end{corollary}

\begin{lemma}
For any initial state $ \left< Q_i, s_0 \right>$ in an induced product Markov Chain $\mathcal{M}_{\otimes}^{\mathcal{A}}$,
\begin{align}
\mathcal{P}_{\mathcal{M}_{\otimes}^{\mathcal{A}}}( \left<Q_i,s_0 \right> \models \Diamond WC) + \mathcal{P}_{\mathcal{M}_{\otimes}^{\mathcal{A}}}(\left< Q_i,s_0 \right> \models \Diamond LC) = 1 \ .
\end{align}
\end{lemma}
\begin{proof}
This lemma follows from Corollary 2.
\end{proof}


\begin{lemma}
Let $\mathcal{I}$ be an IMC and let $\mathcal{M}_1$ and $\mathcal{M}_2$ be two MCs induced by $\mathcal{I}$ where the set $B$ is a BSCC for both. If $C_1$ and $C_2$ are the sets of states such that $\mathcal{P}_{\mathcal{M}_1}( C_1 \models \Diamond  B) = 1$ and $\mathcal{P}_{\mathcal{M}_2}( C_2 \models \Diamond  B) = 1$, then there exists a MC $\mathcal{M}_3$ induced by $\mathcal{I}$ such that $\mathcal{P}_{\mathcal{M}_3}( (C_1 \cup C_2) \models \Diamond  B) = 1$.
\end{lemma}
\begin{proof}
Let $T_1$ and $T_2$ denote the transition matrices of $\mathcal{M}_1$ and $\mathcal{M}_2$ respectively, and $Q$ denote the set of states in $\mathcal{I}$. Consider an induced MC $\mathcal{M}_3$ where $T_3(Q_i, Q_j) = T_1(Q_i, Q_j) \; \forall Q_i \in C_1 $ and $\forall \; Q_j \in Q$, and $T_3(Q_i, Q_j) = T_2(Q_i, Q_j) \; \forall Q_i  \in C_2 \setminus (C_1 \cap C_2 )$ and $\forall \; Q_j \in Q$. By assumption, any state in $C_1$ reaches $B$ with probability 1, while all states in $C_2 \setminus (C_1 \cap C_2 )$ reach $B \cup (C_1 \cap C_2)$ with probability 1. Since $\mathcal{P}_{\mathcal{M}_3}((C_1 \cap C_2) \models \Diamond  B) = 1$ by construction, we have $\mathcal{P}_{\mathcal{M}_3}( (C_1 \cup C_2) \models \Diamond  B) = 1$. 
\end{proof}


\begin{lemma}
Let $\mathcal{I} \otimes \mathcal{A}$ be a product IMC and $(LC)_i$ be the losing components of any product MC $(\mathcal{M}_{\otimes}^{\mathcal{A}})_i$ induced by $\mathcal{I} \otimes \mathcal{A}$. There exists a set of product MCs induced by $\mathcal{I} \otimes  \mathcal{A}$ with losing components $(LC)_L$ and such that $(LC)_i \subseteq (LC)_L$.
\end{lemma}
\begin{proof}
We proved in \cite{dutreixCDC2018} that any product IMC induces a set of MCs with a largest set of non-accepting BSCCs. Lemma 5 is deduced from this fact and Lemma 4.
\end{proof}


\begin{lemma}
Let $\mathcal{I} \otimes \mathcal{A}$ be a product IMC. Let $(\mathcal{M}_{\otimes}^ \mathcal{A})_1$ and $(\mathcal{M}_{\otimes}^ \mathcal{A})_2$ be two product MCs induced by $\mathcal{I} \otimes \mathcal{A}$ with sets of accepting BSCC $U^A_1$ and $U^A_2$ respectively. There exists a set of product MCs induced by $\mathcal{I} \otimes \mathcal{A}$ with winning components $(WC)_3$ and such that $( U^A_1 \cup U^A_2 ) \subseteq (WC)_3$. 
\end{lemma}

\begin{proof}
Let $T_1$ and $T_2$ denote the transition matrices of $(\mathcal{M}_{\otimes}^{\mathcal{A}})_1$ and $(\mathcal{M}_{\otimes}^{\mathcal{A}})_2$ respectively, and $Q$ denote the set of states in $\mathcal{I} \otimes \mathcal{A}$. Assume $U^A_1 \cap U^A_2 = \emptyset$. There exists a set of product MCs induced by $\mathcal{I} \otimes \mathcal{A}$ such that $U^A_1 \cup U^A_2$ are accepting BSCCs (see \cite{dutreixCDC2018}), and thus winning components. If $U^A_1 \cap U^A_2 \not = \emptyset$, consider the set of all product MCs $(\mathcal{M}_{\otimes}^{\mathcal{A}})_i$ induced by $\mathcal{I} \otimes \mathcal{A}$ such that, for all transition matrices $T_i$ of the product MCs in this set, $T_i(Q_i, Q_j) = T_1(Q_i, Q_j) \; \forall Q_i \in U^A_1 $ and $\forall \; Q_j \in Q$, and $T_i(Q_i, Q_j) = T_2(Q_i, Q_j) \; \forall Q_i \in U^A_2 \setminus (U^A_1 \cap U^A_2 )$ and $\forall \; Q_j \in Q$. Clearly, $U^A_1$ is an accepting BSCC in all $(\mathcal{M}_{\otimes}^{\mathcal{A}})_i$. For all $(\mathcal{M}_{\otimes}^{\mathcal{A}})_i$, it holds that $\mathcal{P}_{(\mathcal{M}_{\otimes}^{\mathcal{A}})_i}((U^A_2 \setminus (U^A_1 \cap U^A_2 )) \models \Diamond  (U^A_1 \cap U^A_2)) = 1$, since $U^A_2$ is a BSCC for the same probability assignments in $(\mathcal{M}_{\otimes}^{\mathcal{A}})_2$. Thus, $U^A_1 \cup U^A_2$ are winning components with respect to all $(\mathcal{M}_{\otimes}^{\mathcal{A}})_i$.
\end{proof}

\begin{lemma}
Let $\mathcal{I} \otimes \mathcal{A}$ be a product IMC and $(WC)_i$ be the winning components of any product MC $(\mathcal{M}_{\otimes}^{\mathcal{A}})_i$ induced by $\mathcal{I} \otimes \mathcal{A}$. There exists a set of product MCs induced by $\mathcal{I} \otimes  \mathcal{A}$ with winning component $(WC)_L$ and such that $(WC)_i \subseteq (WC)_L$.
\end{lemma}
\begin{proof}
This lemma follows from Lemmas 4 and 6.
\end{proof}


\begin{lemma}
Let $\mathcal{I} \otimes \mathcal{A}$ be a product IMC. Let $(\mathcal{M}_{\otimes}^ \mathcal{A})_1$ and $(\mathcal{M}_{\otimes}^{\mathcal{A}})_2$ be two product MCs induced by $\mathcal{I} \otimes \mathcal{A}$ with winning components $(WC)_1$ and $(WC)_2$ respectively, and such that $(WC)_2 \subseteq (WC)_1$. Also, their losing components $(LC)_1$ and $(LC)_2$ are such that $(LC)_1 = (LC)_2 = LC$ and their respective transition matrices $T_1$ and $T_2$ satisfy $T_1(Q_i, Q_j) = T_2(Q_i, Q_j)  \; \forall Q_i \in (Q \times S) \setminus ((WC)_1 \cup LC)$ and $\forall Q_j \in (Q \times S)$. The sets of accepting BSCCs of $(\mathcal{M}_{\otimes}^ \mathcal{A})_1$ and $(\mathcal{M}_{\otimes}^ \mathcal{A})_2$ are denoted by $U^A_1$ and $U^A_2$ respectively. For any initial state $\left<Q_i,s_0 \right>$, it holds that
\begin{align*}
\mathcal{P}_{(\mathcal{M}_{\otimes}^ \mathcal{A})_1}(\left<Q_i,s_0 \right> \models \Diamond U^A_1) \geq \mathcal{P}_{(\mathcal{M}_{\otimes}^{\mathcal{A}})_2}(\left<Q_i,s_0 \right> \models \Diamond U^A_2) \; .
\end{align*} 
\end{lemma}
\begin{proof}
For any initial state $\left<Q_i,s_0 \right> \in LC $, it holds that $\mathcal{P}_{(\mathcal{M}_{\otimes}^ \mathcal{A})_1}( \left<Q_i,s_0 \right> \models \Diamond U^A_1) = \mathcal{P}_{(\mathcal{M}_{\otimes}^ \mathcal{A})_2}(\left<Q_i,s_0 \right> \models \Diamond U^A_2) = 0$. For any initial state $\left<Q_i,s_0 \right> \in ((WC)_1 \cap (WC)_2)$, it holds that $\mathcal{P}_{(\mathcal{M}_{\otimes}^ \mathcal{A})_1}(\left<Q_i,s_0 \right> \models \Diamond U^A_1) = \mathcal{P}_{(\mathcal{M}_{\otimes}^ \mathcal{A})_2}(\left<Q_i,s_0 \right> \models \Diamond U^A_2) = 1$. For any initial state $\left<Q_i,s_0 \right> \in ((WC)_1 \setminus (WC)_2)$, it holds that $\mathcal{P}_{(\mathcal{M}_{\otimes}^ \mathcal{A})_1}(\left<Q_i,s_0 \right> \models \Diamond U^A_1) = 1 \geq \mathcal{P}_{(\mathcal{M}_{\otimes}^ \mathcal{A})_2}(\left<Q_i,s_0 \right> \models \Diamond U^A_2)$. For any initial state $\left<Q_i,s_0 \right> \in (Q \times S) \setminus ((WC)_1 \cup LC)$ (denoted by $H$ for clarity), we have
\begin{align*}
& \mathcal{P}_{(\mathcal{M}_{\otimes}^ \mathcal{A})_1}(\left<Q_i,s_0 \right> \models \Diamond U^A_1) = \\
& T_{1}(\left<Q_i,s_0 \right>, (WC)_1 \setminus (WC)_2) \cdot \mathcal{P}_{(\mathcal{M}_{\otimes}^ \mathcal{A})_1}(\; (WC)_1 \setminus (WC)_2  \models  \\
& \Diamond U^A_1 \;) + T_{1}(\left<Q_i,s_0 \right>, (WC)_2) \cdot \mathcal{P}_{(\mathcal{M}_{\otimes}^ \mathcal{A})_1}(\; (WC)_2  \models \Diamond U^A_1)\\
& + \sum_{Q_j \in H}T_{1}(\left<Q_i,s_0 \right>, Q_j) \cdot \mathcal{P}_{(\mathcal{M}_{\otimes}^ \mathcal{A})_2}(\; Q_j  \models \Diamond U^A_2 \;) \\
& \geq \\
&  \sum_{Q_j \in (WC)_1 \setminus (WC)_2}T_{2}(\left<Q_i,s_0 \right>, Q_j) \cdot \mathcal{P}_{(\mathcal{M}_{\otimes}^ \mathcal{A})_2}(\; Q_j  \models \Diamond U^A_2 \;) \\
& + T_{2}(\left<Q_i,s_0 \right>, (WC)_2) \cdot \mathcal{P}_{(\mathcal{M}_{\otimes}^ \mathcal{A})_2}(\; (WC)_2  \models \Diamond U^A_2)\\
& + \sum_{Q_j \in H}T_{2}(\left<Q_i,s_0 \right>, Q_j) \cdot \mathcal{P}_{(\mathcal{M}_{\otimes}^ \mathcal{A})_2}(\; Q_j  \models \Diamond U^A_2 \;) \\
& = \mathcal{P}_{(\mathcal{M}_{\otimes}^ \mathcal{A})_2}(\left<Q_i,s_0 \right> \models \Diamond U^A_2)
\end{align*}
\noindent based on the transition matrices assumptions.
\end{proof}\mbox{}

\begin{lemma}
Let $\mathcal{I} \otimes \mathcal{A}$ be a product IMC. Let $(\mathcal{M}_{\otimes}^ \mathcal{A})_1$ and $(\mathcal{M}_{\otimes}^{\mathcal{A}})_2$ be two product MCs induced by $\mathcal{I} \otimes \mathcal{A}$ with losing components $(LC)_1$ and $(LC)_2$ respectively, and such that $(LC)_2 \subseteq (LC)_1$. Also, their winning components $(WC)_1$ and $(WC)_2$ are such that $(WC)_1 = (WC)_2 = WC$ and their respective transition matrices $T_1$ and $T_2$ satisfy $T_1(Q_i, Q_j) = T_2(Q_i, Q_j)  \; \forall Q_i \in (Q \times S) \setminus ((LC)_1 \cup WC)$ and $\forall Q_j \in (Q \times S)$. The sets of non-accepting BSCCs of $(\mathcal{M}_{\otimes}^ \mathcal{A})_1$ and $(\mathcal{M}_{\otimes}^ \mathcal{A})_2$ are denoted by $U^N_1$ and $U^N_2$ respectively. For any initial state $\left<Q_i,s_0 \right>$, it holds that
\begin{align*}
\mathcal{P}_{(\mathcal{M}_{\otimes}^ \mathcal{A})_1}(\left<Q_i,s_0 \right> \models \Diamond U^N_1) \geq \mathcal{P}_{(\mathcal{M}_{\otimes}^{\mathcal{A}})_2}(\left<Q_i,s_0 \right> \models \Diamond U^N_2) \;.
\end{align*} 
\end{lemma}
\begin{proof}
The proof is identical to the one of Lemma 8
\end{proof}

\begin{lemma}
Let $\mathcal{I} \otimes \mathcal{A}$ be a product IMC with permanent and largest sets $(WC)_{P}$, $(LC)_{P}$, $(WC)_{L}$ and $(LC)_{L}$ as previously defined. There exists a set of induced MCs of $\mathcal{I} \otimes \mathcal{A}$ whose sets of winning and losing components are $(WC)_{P}$ and $(LC)_L$, and a set of induced MCs whose sets of losing and winning components are $(LC)_P$ and $(WC)_{L}$.
\end{lemma}
\begin{proof}
Consider the set $\mathcal{C}$ of all induced MCs of $\mathcal{I} \otimes \mathcal{A}$ whose set of losing components is $(LC)_{L}$. For any $\left<Q_i, s_0 \right> \in ((WC)_? \setminus (LC)_{L})$, consider an induced product MC $\mathcal{M} \in \mathcal{C}$ with transition matrix $T$ such that $\left<Q_i, s_0 \right>$ is not a winning component of $\mathcal{M}$. Such an induced product MC always exists by the definition of $(WC)_?$ and Lemma 8. Denote by $(WC)_?^{\mathcal{M}}$ the winning components of $\mathcal{M}$ which also belong to $(WC)_?$. There exists an induced product MC $\mathcal{M'}$ with transition matrix $T'$ such that, for all $q_{i} \in (WC)_?^{\mathcal{M}}$, $\mathcal{P}_{\mathcal{M'}}(q_{i} \models \Diamond \neg ((WC)_?^{\mathcal{M}} \cup (WC)_P)) > 0$, otherwise $q_{i} \in (WC)_{P}$, which is a contradiction. Consider the induced product MC $\mathcal{M''} \in \mathcal{C}$ with transition matrix $T''$ such that $T''(q_{i}, q_{j}) = T'(q_{i}, q_{j})$ for all $q_{i} \in (WC)_?^{\mathcal{M}}$ and $q_{j} \in (Q \times S)$,  and $T'' = T$ for all other transitions. The sets of winning and losing components of $\mathcal{M''}$ are $(WC)_{P}$ and $(LC)_{L}$, proving the claim.
The proof with respect to $(LC)_{P}$ and $(WC)_{L}$ is identical.
\end{proof}

\bibliographystyle{IEEEtran}
\bibliography{TAC1.bib}

\begin{thebibliography}{10}
\providecommand{\url}[1]{#1}
\csname url@samestyle\endcsname
\providecommand{\newblock}{\relax}
\providecommand{\bibinfo}[2]{#2}
\providecommand{\BIBentrySTDinterwordspacing}{\spaceskip=0pt\relax}
\providecommand{\BIBentryALTinterwordstretchfactor}{4}
\providecommand{\BIBentryALTinterwordspacing}{\spaceskip=\fontdimen2\font plus
\BIBentryALTinterwordstretchfactor\fontdimen3\font minus
  \fontdimen4\font\relax}
\providecommand{\BIBforeignlanguage}[2]{{%
\expandafter\ifx\csname l@#1\endcsname\relax
\typeout{** WARNING: IEEEtran.bst: No hyphenation pattern has been}%
\typeout{** loaded for the language `#1'. Using the pattern for}%
\typeout{** the default language instead.}%
\else
\language=\csname l@#1\endcsname
\fi
#2}}
\providecommand{\BIBdecl}{\relax}
\BIBdecl

\bibitem{baier2008principles}
C.~Baier, J.-P. Katoen, and K.~G. Larsen, \emph{Principles of model
  checking}.\hskip 1em plus 0.5em minus 0.4em\relax MIT press, 2008.

\bibitem{aziz2000model}
A.~Aziz, K.~Sanwal, V.~Singhal, and R.~Brayton, ``Model-checking
  continuous-time markov chains,'' \emph{ACM Transactions on Computational
  Logic (TOCL)}, vol.~1, no.~1, pp. 162--170, 2000.

\bibitem{baier2003model}
C.~Baier, B.~Haverkort, H.~Hermanns, and J.-P. Katoen, ``{Model-checking
  algorithms for continuous-time Markov chains},'' \emph{IEEE Transactions on
  software engineering}, vol.~29, no.~6, pp. 524--541, 2003.

\bibitem{kwiatkowska2011prism}
M.~Kwiatkowska, G.~Norman, and D.~Parker, ``Prism 4.0: Verification of
  probabilistic real-time systems,'' in \emph{International conference on
  computer aided verification}.\hskip 1em plus 0.5em minus 0.4em\relax
  Springer, 2011, pp. 585--591.

\bibitem{abate2011approximate}
A.~Abate, A.~D'Innocenzo, and M.~D. Di~Benedetto, ``{Approximate abstractions
  of stochastic hybrid systems},'' \emph{IEEE Transactions on Automatic
  Control}, vol.~56, no.~11, pp. 2688--2694, 2011.

\bibitem{soudjani2014faust}
S.~Soudjani, C.~Gevaerts, and A.~Abate, ``Faust 2: Formal abstractions of
  uncountable-state stochastic processes,'' \emph{arXiv preprint
  arXiv:1403.3286}, 2014.

\bibitem{abate2008markov}
A.~Abate, A.~D'Innocenzo, M.~D. Di~Benedetto, and S.~S. Sastry, ``{Markov
  set-chains as abstractions of stochastic hybrid systems},'' in
  \emph{International Workshop on Hybrid Systems: Computation and
  Control}.\hskip 1em plus 0.5em minus 0.4em\relax Springer, 2008, pp. 1--15.

\bibitem{abate2011quantitative}
A.~Abate, J.-P. Katoen, and A.~Mereacre, ``Quantitative automata model checking
  of autonomous stochastic hybrid systems,'' in \emph{Proceedings of the 14th
  international conference on Hybrid systems: computation and control}.\hskip
  1em plus 0.5em minus 0.4em\relax ACM, 2011, pp. 83--92.

\bibitem{tkachev2013formula}
I.~Tkachev and A.~Abate, ``Formula-free finite abstractions for linear temporal
  verification of stochastic hybrid systems,'' in \emph{Proceedings of the 16th
  international conference on Hybrid systems: computation and control}.\hskip
  1em plus 0.5em minus 0.4em\relax ACM, 2013, pp. 283--292.

\bibitem{hahn2013compositional}
E.~M. Hahn, A.~Hartmanns, H.~Hermanns, and J.-P. Katoen, ``A compositional
  modelling and analysis framework for stochastic hybrid systems,''
  \emph{Formal Methods in System Design}, vol.~43, no.~2, pp. 191--232, 2013.

\bibitem{franzle2011measurability}
M.~Fr{\"a}nzle, E.~M. Hahn, H.~Hermanns, N.~Wolovick, and L.~Zhang,
  ``Measurability and safety verification for stochastic hybrid systems,'' in
  \emph{Proceedings of the 14th international conference on Hybrid systems:
  computation and control}.\hskip 1em plus 0.5em minus 0.4em\relax ACM, 2011,
  pp. 43--52.

\bibitem{lahijanian2015formal}
M.~Lahijanian, S.~B. Andersson, and C.~Belta, ``Formal verification and
  synthesis for discrete-time stochastic systems,'' \emph{IEEE Transactions on
  Automatic Control}, vol.~60, no.~8, pp. 2031--2045, 2015.

\bibitem{dutreix2018}
M.~Dutreix and S.~Coogan, ``{Efficient Verification for Stochastic Mixed
  Monotone Systems},'' in \emph{International Conference on Cyber-Physical
  Systems}, 2018.

\bibitem{kozine2002interval}
I.~O. Kozine and L.~V. Utkin, ``{Interval-valued finite Markov chains},''
  \emph{Reliable computing}, vol.~8, no.~2, pp. 97--113, 2002.

\bibitem{smith2008global}
H.~Smith, ``Global stability for mixed monotone systems,'' \emph{Journal of
  Difference Equations and Applications}, vol.~14, no. 10-11, pp. 1159--1164,
  2008.

\bibitem{coogan2015efficient}
S.~Coogan and M.~Arcak, ``Efficient finite abstraction of mixed monotone
  systems,'' in \emph{Proceedings of the 18th International Conference on
  Hybrid Systems: Computation and Control}.\hskip 1em plus 0.5em minus
  0.4em\relax ACM, 2015, pp. 58--67.

\bibitem{Hirsch:1985fk}
M.~W. Hirsch, ``Systems of differential equations that are competitive or
  cooperative {II}: Convergence almost everywhere,'' \emph{SIAM Journal on
  Mathematical Analysis}, vol.~16, no.~3, pp. 423--439, 1985.

\bibitem{Smith:2008fk}
H.~L. Smith, \emph{Monotone dynamical systems: {A}n introduction to the theory
  of competitive and cooperative systems}.\hskip 1em plus 0.5em minus
  0.4em\relax American Mathematical Society, 1995.

\bibitem{Angeli:2003fv}
D.~Angeli and E.~Sontag, ``Monotone control systems,'' \emph{IEEE Transactions
  on Automatic Control}, vol.~48, no.~10, pp. 1684--1698, 2003.

\bibitem{Sontag:2007ad}
E.~D. Sontag, ``Monotone and near-monotone biochemical networks,''
  \emph{Systems and Synthetic Biology}, vol.~1, no.~2, pp. 59--87, 2007.

\bibitem{Gomes:2008fk}
\BIBentryALTinterwordspacing
G.~Gomes, R.~Horowitz, A.~A. Kurzhanskiy, P.~Varaiya, and J.~Kwon, ``Behavior
  of the cell transmission model and effectiveness of ramp metering,''
  \emph{Transportation Research Part C: Emerging Technologies}, vol.~16, no.~4,
  pp. 485--513, 2008. [Online]. Available:
  \url{http://www.sciencedirect.com/science/article/pii/S0968090X0700085X}
\BIBentrySTDinterwordspacing

\bibitem{Lovisari:2014yq}
E.~Lovisari, G.~Como, and K.~Savla, ``Stability of monotone dynamical flow
  networks,'' in \emph{Proceedings of the 53rd Conference on Decision and
  Control}, 2014, pp. 2384--2389.

\bibitem{Coogan:2016rp}
S.~Coogan, M.~Arcak, and A.~A. Kurzhanskiy, ``Mixed monotonicity of partial
  first-in-first-out traffic flow models,'' in \emph{IEEE Conference on
  Decision and Control}, 2016, pp. 7611--7616.

\bibitem{rozier2011linear}
K.~Y. Rozier, ``Linear temporal logic symbolic model checking,'' \emph{Computer
  Science Review}, vol.~5, no.~2, pp. 163--203, 2011.

\bibitem{baier2004controller}
C.~Baier, M.~Gr{\"o}{\ss}er, M.~Leucker, B.~Bollig, and F.~Ciesinski,
  ``Controller synthesis for probabilistic systems,'' in \emph{Exploring New
  Frontiers of Theoretical Informatics}.\hskip 1em plus 0.5em minus 0.4em\relax
  Springer, 2004, pp. 493--506.

\bibitem{dutreixCDC2018}
M.~Dutreix and S.~Coogan, ``{Satisfiability Bounds for $\omega$-regular
  Properties in Interval-valued Markov Chains},'' in \emph{Proceedings of the
  57th IEEE Conference on Decision and Control}.\hskip 1em plus 0.5em minus
  0.4em\relax (to appear), 2018.

\bibitem{sen2006model}
K.~Sen, M.~Viswanathan, and G.~Agha, ``{Model-checking Markov chains in the
  presence of uncertainties},'' in \emph{International Conference on Tools and
  Algorithms for the Construction and Analysis of Systems}.\hskip 1em plus
  0.5em minus 0.4em\relax Springer, 2006, pp. 394--410.

\bibitem{Hirsch:2005ek}
M.~Hirsch and H.~Smith, ``Monotone dynamical systems,'' \emph{Handbook of
  differential equations: {O}rdinary differential equations}, vol.~2, pp.
  239--357, 2005.

\bibitem{cauchi2019efficiency}
N.~Cauchi, L.~Laurenti, M.~Lahijanian, A.~Abate, M.~Kwiatkowska, and
  L.~Cardelli, ``Efficiency through uncertainty: scalable formal synthesis for
  stochastic hybrid systems,'' in \emph{Proceedings of the 22nd ACM
  International Conference on Hybrid Systems: Computation and Control}.\hskip
  1em plus 0.5em minus 0.4em\relax ACM, 2019, pp. 240--251.

\bibitem{klein2006experiments}
J.~Klein and C.~Baier, ``Experiments with deterministic $\omega$-automata for
  formulas of linear temporal logic,'' \emph{Theoretical Computer Science},
  vol. 363, no.~2, pp. 182--195, 2006.

\bibitem{babiak2013effective}
T.~Babiak, F.~Blahoudek, M.~K{\v{r}}et{\'\i}nsk{\`y}, and J.~Strej{\v{c}}ek,
  ``{Effective translation of LTL to deterministic Rabin automata: Beyond the
  (F, G)-fragment},'' in \emph{International Symposium on Automated Technology
  for Verification and Analysis}.\hskip 1em plus 0.5em minus 0.4em\relax
  Springer, 2013, pp. 24--39.

\bibitem{chen2013complexity}
T.~Chen, T.~Han, and M.~Kwiatkowska, ``{On the complexity of model checking
  interval-valued discrete time Markov chains},'' \emph{Information Processing
  Letters}, vol. 113, no.~7, pp. 210--216, 2013.

\bibitem{chatterjee2008model}
K.~Chatterjee, K.~Sen, and T.~Henzinger, ``Model-checking $\omega$-regular
  properties of interval markov chains,'' \emph{Foundations of Software Science
  and Computational Structures}, pp. 302--317, 2008.

\bibitem{de1999computing}
L.~De~Alfaro, ``Computing minimum and maximum reachability times in
  probabilistic systems,'' in \emph{International Conference on Concurrency
  Theory}.\hskip 1em plus 0.5em minus 0.4em\relax Springer, 1999, pp. 66--81.

\bibitem{cauchi2019stochy}
N.~Cauchi, K.~Degiorgio, and A.~Abate, ``Stochy: automated verification and
  synthesis of stochastic processes,'' \emph{arXiv preprint arXiv:1901.10287},
  2019.

\end{thebibliography}

\begin{IEEEbiography}[{\includegraphics[width=1in,height=1.25in,clip,keepaspectratio]{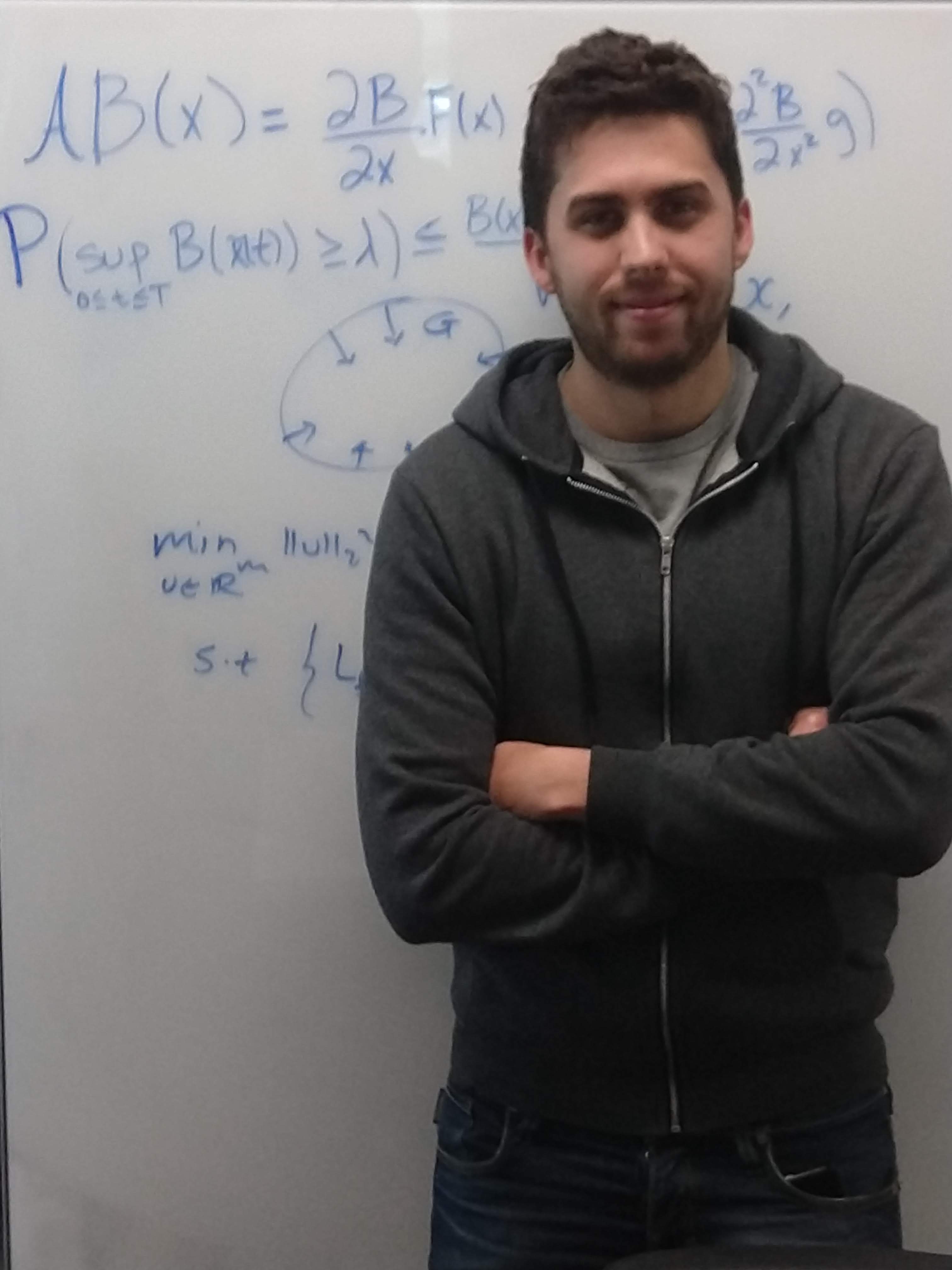}}]{Maxence Dutreix} received the B.S. degree in Engineering Physics from the University of California, San Diego, CA, USA, in 2015.
He is currently a Ph.D. student in Electrical and Computer Engineering at the Georgia Institute of Technology, Atlanta, GA, USA. His research interest is in the area of verification and synthesis for stochastic dynamical systems.
\end{IEEEbiography}

\begin{IEEEbiography}[{\includegraphics[width=1in,height=1.25in,clip,keepaspectratio]{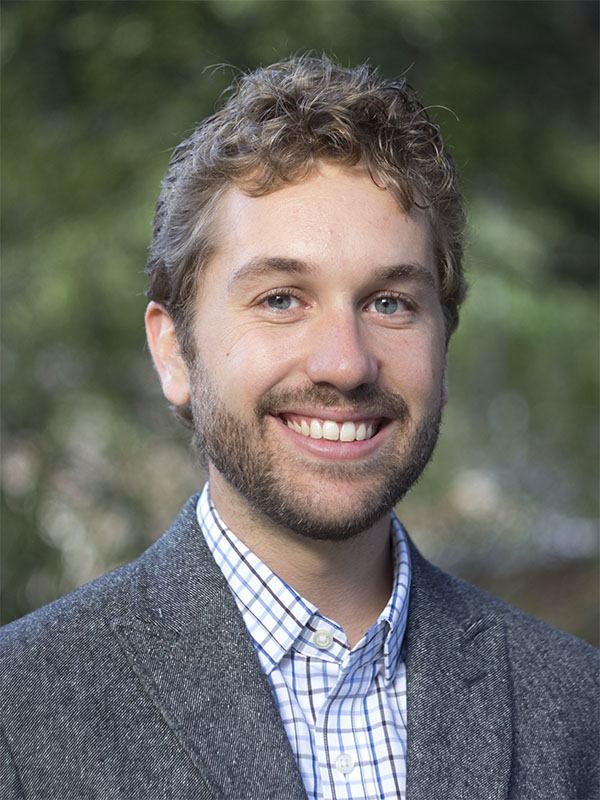}}]{Samuel Coogan} is an Assistant Professor at Georgia Tech in the School of Electrical and Computer Engineering and the School of Civil and Environmental Engineering. Prior to joining Georgia Tech in July 2017, he was an assistant professor in the Electrical Engineering Department at UCLA from 2015-2017. He received the B.S. degree in Electrical Engineering from Georgia Tech and the M.S. and Ph.D. degrees in Electrical Engineering from the University of California, Berkeley. His research is in the area of dynamical systems and autonomy and focuses on developing scalable tools for verification and control of networked, cyber-physical systems with an emphasis on transportation systems. He received a CAREER award from NSF in 2018, a Young Investigator Award from the Air Force Office of Scientific Research in 2018, and the Outstanding Paper Award for the IEEE Transactions on Control of Network Systems in 2017.
\end{IEEEbiography}

\end{document}